\documentclass{llncs}
\pdfoutput=1
\usepackage{amssymb}
\usepackage{amsmath}
\usepackage{latexsym}
\usepackage{textcomp}
\usepackage{epsfig}
\usepackage{xspace}
\usepackage{todonotes}
\usepackage{paralist} 

\usepackage{wrapfig}

\usepackage[labelformat=simple]{subcaption}
\graphicspath{{figures/}}

\spnewtheorem*{sketch}{Sketch of proof}{\itshape}{\rmfamily}

\newcommand{\remove}[1]{}

\DeclareMathOperator{\Ham}{\mathsf{Ham}}%
\DeclareMathOperator{\aug}{\mathsf{aug}}%
\DeclareMathOperator{\sub}{\mathsf{sub}}%
\newcommand{\col}{\mathrm{col}}
\newcommand{\C}{\mathcal C}%
\newcommand{\Pa}{\mathcal P}%

\begin{document}
\title{Colored Point-set Embeddings of Acyclic Graphs\thanks{The work has been supported in part by 
the European project  ``Geospatial based Environment for Optimisation Systems Addressing Fire 
Emergencies'' (GEO-SAFE), contract no. H2020-691161, by the Network Sciences and Technologies (NeST) initiative at University of Liverpool, and by 
the Italian project: ``RISE: un nuovo framework distribuito per data collection, monitoraggio e comunicazioni in contesti di emergency response'', Fondazione Cassa  Risparmio Perugia, code 2016.0104.021.}
}

\author{
Emilio Di Giacomo\inst{1},
Leszek Gasieniec\inst{2},
Giuseppe Liotta\inst{1},
Alfredo Navarra\inst{1}
}

\date{}

\institute{
Universit\`a degli Studi di Perugia, Italy\\
\email{\{emilio.digiacomo,giuseppe.liotta,alfredo.navarra\}@unipg.it}
\and
University of Liverpool, UK\\
\email{L.A.Gasieniec@liverpool.ac.uk}
}

\maketitle

\begin{abstract}
We show that any planar drawing of a forest of three stars whose vertices are constrained to be at fixed vertex locations may require $\Omega(n^\frac{2}{3})$ edges each having $\Omega(n^\frac{1}{3})$ bends in the worst case. The lower bound holds even when the function that maps vertices to points is not a bijection but it is defined by a 3-coloring. In contrast, a constant number of bends per edge can be obtained for 3-colored paths and for 3-colored caterpillars whose leaves all have the same color. Such results answer to a long standing open problem.
\end{abstract}

\section{Introduction}

A pioneering paper by Pach and Wenger~\cite{DBLP:journals/gc/PachW01} studied the problem of computing a planar drawing of a graph $G$ with the constraint that the mapping of the vertices to the points in the plane, that represent the vertices of $G$, is given as part of the input. Pach and Wenger proved that, for any given mapping, a planar graph with $n$ vertices admits a planar drawing such that the curve complexity, i.e. the number of bends per edge, is $O(n)$. Furthermore, they proved that the bound on the curve complexity is (almost surely) tight as $n$ tends to infinity when $G$ has $O(n)$ independent edges. This implies that the curve complexity of a planar drawing with vertices at fixed locations may be $\Omega(n)$ even for structurally very simple graphs such as paths or matchings, for which the number of independent edges is linear in $n$.

These results have motivated the study of a relaxed version of the problem where the function that associates vertices of the graph to points of the plane is not a bijection. Namely, an instance of the $k$-colored point set embeddability problem receives as input an $n$-vertex planar graph $G$ such that every vertex is given one of $k$ distinct colors and a set $S$ of $n$ distinct points, such that, each point is given one of the $k$ distinct colors. The number of points of $S$ having a certain color $i$ is the same as the number of vertices of $G$ having color $i$. The goal is to compute a planar drawing of $G$ with curve complexity independent of $n$ where every vertex of a specific color is represented by a point of the same color. When $k=n$ the $k$-colored point set embeddability problem coincides with the problem of computing a drawing with vertices at fixed locations and thus the lower bounds by Pach and Wenger hold. Therefore several papers have focused on small values of $k$ (typically $k \leq 3$) to see whether better bounds on the curve complexity could be achieved in this scenario (see, e.g.,~\cite{DBLP:journals/tcs/BadentGL08,DBLP:conf/walcom/GiacomoL10,DBLP:journals/ijfcs/GiacomoLT06,DBLP:journals/algorithmica/GiacomoLT10,DBLP:conf/gd/FratiGLLMN12}).

For $k=1$, Kaufmann and Wiese~\cite{DBLP:journals/jgaa/KaufmannW02} proved that every planar graph admits a $1$-colored point set embedding onto any point set with curve complexity at most $2$. For $k = 2$, outerplanar graphs always admit a $2$-colored point set embedding with $O(1)$ curve complexity~\cite{DBLP:journals/jgaa/GiacomoDLMTW08}. However, for $k \geq 2$, there are $2$-connected $k$-colored planar graphs for which a $k$-colored point set embedding may require $\Omega(n)$ bends on $\Omega(n)$ edges~\cite{DBLP:journals/tcs/BadentGL08}. This result extends the lower bound of Pach and Wenger~\cite{DBLP:journals/gc/PachW01} to a much more relaxed set of constraints on the location of the vertices, but it does so by using 2-connected graphs instead of just (not necessarily connected) planar graphs. For example, the problem of establishing tight bounds on the curve complexity of $k$-colored forests for small values of $k \geq 3$ is a long standing open problem (see, e.g.,~\cite{DBLP:journals/tcs/BadentGL08}).
We explicitly address this gap in the literature and consider the $k$-colored point set embeddability problem for acyclic graphs and $k \geq 3$. Our main results are as follows.

\begin{itemize}

\item In Section~\ref{se:lower}, we prove that a planar drawing of a forest of three stars and $n$ vertices may require $\Omega(n^\frac{2}{3})$ edges with $\Omega(n^\frac{1}{3})$ bends each, even if the mapping of the vertices to the points is defined by using a set of $k$ colors with $k\geq 3$. In contrast, a constant number of bends per edge can always be achieved if the number of stars is at most two (for any number of colors) or the number of colors  is at most two (for any number of stars).
    
\item Since the above result implies that $3$-colored point set embeddings of $3$-colored caterpillars may have a non-constant curve complexity, in Section~\ref{se:upper} we study subfamilies of $3$-colored caterpillars for which constant curve complexity is possible. We prove that every $3$-colored path and every $3$-colored caterpillar whose leaves all have the same color admit a $3$-colored point-set emebdding with constant curve complexity onto any $3$-colored point set. 

\item Finally, still in Section~\ref{se:upper}, we 
prove that any $4$-colored path $\pi$ such that the vertices of colors $0$ and $1$ precede all vertices of colors $2$ and $3$ when moving along $\pi$ has a $4$-colored point set embedding with at most five bends per edge onto any $4$-colored point-set. 

\end{itemize}

Concerning the lower bound, it is worth mentioning that the argument by Pach and Wenger~\cite{DBLP:journals/gc/PachW01}  does not apply to families of graphs where the number of independent edges is not a function of $n$. Hence, our lower bound extends the one by Pach and Wenger about the curve complexity of planar drawings with vertices at fixed locations also to those graphs for which the number of independent edges does not grow with $n$.
Some proofs can be found in the appendix.

\section{Preliminaries}\label{se:preliminaries}
Let $G=(V,E)$ be a graph. A \emph{$k$-coloring} of $G$ is a
partition $\{V_0, V_1, \dots, V_{k-1}\}$ of $V$. The integers
$0, 1, \dots, k-1$ are called \emph{colors} and $G$ is called a \emph{$k$-colored graph}.  For each vertex $v \in V_i$ we denote by $col(v)$ the
color $i$ of $v$. 

Let $S$ be a set of distinct points in the plane. For any point $p
\in S$, we denote by $x(p)$ and $y(p)$ the $x$- and $y$-coordinates
of $p$, respectively. We denote by $CH(S)$ the convex hull of $S$. Throughout the paper we
always assume that the points of $S$ have different
$x$-coordinates (if not we can rotate the plane so to achieve this
condition). 
A \emph{$k$-coloring} of $S$ is a partition
$\{S_0, S_1, \dots, S_{k-1}\}$ of $S$. A set of points $S$ with a
$k$-coloring is called a \emph{$k$-colored point set}. For each point $p \in
S_i$, $col(p)$ denotes the color $i$ of $p$. A $k$-colored point set $S$ is \emph{compatible with} a $k$-colored graph $G$ if $|V_i| = |S_i|$ for every $i$. If $G$ is planar we say that $G$
has a \emph{topological point-set embedding} on $S$ if there exists a planar drawing of $G$ such that: (i) every vertex $v$ is mapped to a distinct
point $p$ of $S$ with $col(p)=col(v)$, (ii) each edge $e$ of $G$
is drawn as simple Jordan arc. We say that $G$ has a \emph{$k$-colored point-set embedding} on $S$ if there exists a planar drawing
of $G$ such that: (i) every vertex $v$ is mapped to a distinct
point $p$ of $S$ with $col(p)=col(v)$, (ii) each edge $e$ of $G$
is drawn as a polyline $\lambda$. A point shared by any two
consecutive segments of $\lambda$ is called a \emph{bend} of $e$. The maximum number of bends along an edge is the \emph{curve complexity} of the $k$-colored point-set embedding.
A \emph{$k$-colored sequence} $\sigma$ is a sequence of (possibly
repeated) colors $c_0$, $c_1$, $\dots$, $c_{n-1}$ such that $0
\leq c_j \leq k-1$ ($0 \leq j \leq n-1$).  We say that $\sigma$ is
\emph{compatible with} a $k$-colored graph $G$ if color $i$ occurs $|V_i|$
times in $\sigma$. Let $S$ be a $k$-colored point set.  
Let $p_0, \dots, p_{n-1}$ be the points of $S$ with  $x(p_0) < \ldots < x(p_{n-1})$. The $k$-colored sequence $col(p_0), \dots col(p_{n-1})$ is called the
\emph{$k$-colored sequence induced by $S$}, and is  denoted as $seq(S)$.
A set of points $S$ is \emph{one-sided convex} if they are in convex position and the two points with minimum and maximum $x$-coordinate are consecutive along $CH(S)$. In a $k$-colored one-sided convex point set, the sequence of colors encountered clockwise along $CH(S)$, starting from the point with minimum $x$-coordinate, coincides with $seq(S)$.

A \emph{Hamiltonian cycle} of a graph $G$ is a simple
cycle that contains all vertices of $G$. A graph $G$ that admits a
Hamiltonian cycle is said to be \emph{Hamiltonian}. A planar graph
$G$ is \emph{sub-Hamiltonian} if either $G$ is Hamiltonian or $G$
can be augmented with dummy edges (but not with dummy vertices) to a
graph $\aug(G)$ that is Hamiltonian and planar. A \emph{subdivision}
of a graph $G$ is a graph obtained from $G$ by replacing each edge
by a path with at least one edge. Internal vertices on such a path
are called \emph{division vertices}. Every planar graph has a subdivision that is
sub-Hamiltonian. Let $G$ be a planar graph and let $\sub(G)$ be a sub-Hamiltonian
subdivision of $G$. The graph $\aug(\sub(G))$ is called a
\emph{Hamiltonian augmentation of $G$} and will be denoted as
$\Ham(G)$. Let $\C$ be the Hamiltonian cycle of a Hamiltonian augmentation
$\Ham(G)$ of $G$. Let $e$ be an edge of $\C$, let $\Pa = \C
\setminus e$ be the Hamiltonian path obtained by removing $e$ from
$\C$, and let $v_0,v_1,\dots,v_{n'-1}$ be the vertices of $G$ in the
order they appear along $\Pa$. Finally, let $\sigma=c_0,c_1,\dots, c_{n'-1}$ be a $k$-colored sequence. $\Pa$ is a \emph{$k$-colored
Hamiltonian path consistent with $\sigma$} if $\col(v_i)=c_i$ ($0
\leq i \leq n'-1$). $\C$ is a \emph{$k$-colored Hamiltonian cycle
consistent with $\sigma$} if there exists an edge $e \in \C$ such
that $\Pa = \C \setminus e$ is a $k$-colored Hamiltonian path
consistent with $\sigma$. $\Ham(G)$ is called a \emph{$k$-colored
Hamiltonian augmentation of $G$ consistent with $\sigma$}. The following theorem has been proved in~\cite{DBLP:journals/jgaa/GiacomoDLMTW08} (see also~\cite{DBLP:journals/tcs/BadentGL08,DBLP:journals/algorithmica/GiacomoLT10}).

\begin{theorem}\emph{\cite{DBLP:journals/jgaa/GiacomoDLMTW08}}\label{th:hamiltonian}
Let $G$ be a $k$-colored planar graph and $S$ be a $k$-colored point set consistent with $G$. If $G$ has a $k$-colored Hamiltonian augmentation consistent with $seq(S)$ and at most $d$ division vertices per edge then $G$ admits a $k$-colored point-set embedding on $S$ with at most $2d+1$ bends per edge.
\end{theorem}

The next lemma can be easily derived from Theorem~\ref{th:hamiltonian}. 

\begin{lemma}\label{le:convex}
Let $G$ be a $k$-colored graph, and $S$ be a $k$-colored one-sided convex point set compatible with $G$. If $G$ has a topological $k$-colored point-set embedding on $S$ such that each edge crosses $CH(S)$ at most $b$ times, then $G$ admits a $k$-colored point-set embedding on $S$ with at most $2b+1$ bends per edge.  
\end{lemma}

Let $G=(V,E)$ be a planar graph. A \emph{topological book embedding} of $G$ is a planar drawing such that all vertices of $G$ are represented as points of a horizontal line  $\ell$, called the \emph{spine}. Each of the half-planes defined by $\ell$ is a \emph{page}. Each edge of a topological book embedding is either in the top page, or completely in the bottom page, or it can be on both pages, in which case it crosses the spine. Each crossing between an edge and the spine is called a \emph{spine crossing}. It is also assumed that in a topological book embedding every edge consists of one or more circular arcs, such that no two consecutive arcs are in the same page\footnote{The more general concept of \emph{$h$-page topological book embedding} exists, where each arc can be drawn on one among $h$ different pages. For simplicity we use the term topological book embedding to mean $2$-page topological book embedding.}. 
Let $G$ be a $k$-colored graph and let $\sigma$ be a $k$-colored sequence compatible with $G$. A topological book embedding of $G$ is consistent with $\sigma$ if the sequence of vertex colors along the spine coincides with $\sigma$. Let $S$ be a $k$-colored point set compatible with a $k$-colored planar graph $G$ and let $seq(S)$ be the $k$-colored sequence induced by $S$. The following lemma can be proved similarly to Lemma~\ref{le:convex}.

\begin{lemma}\label{le:top-pse}
If $G$ admits a topological book embedding consistent with $seq(S)$ and having at most $h$ spine crossing per edge, then $G$ admits a point-set embedding on $S$ with curve complexity at most $2h +1$.
\end{lemma}

\section{Point-set Embeddings of Stars}\label{se:lower}

In this section we establish that a $3$-colored point-set embedding of a forest of three stars may require $\Omega(n^{\frac{1}{3}})$ bends along $\Omega(n^{\frac{2}{3}})$ edges by exploiting a previous result about biconnected outerplanar graphs.  We start by recalling the result in~\cite{DBLP:journals/jgaa/GiacomoDLMTW08}.
An {\em alternating point set} $S_n$ is a $3$-colored one-sided convex point set such that: (i) $S_n$ has $n$ points for each color $0$, $1$, and $2$, and (ii) when going along the convex hull $CH(S_n)$ of $S_n$ in clockwise order, the sequence of colors encountered is $0,1,2,0,1,2,\ldots$. Each set of consecutive points colored $0,1,2$ is called a \emph{triplet}. 

\setlength\intextsep{0pt}
\setlength{\columnsep}{5pt}%
\begin{wrapfigure}{I}{0.4\textwidth}
	\centering
	\includegraphics[width=0.24\textwidth]{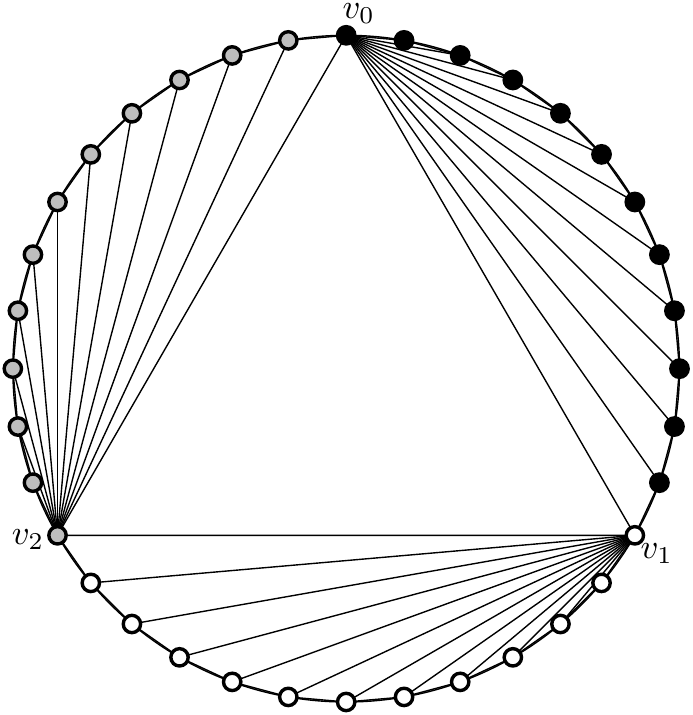}
	\caption{\label{fi:Gn}A $3$-fan $G_n$ for $n=12$.}
\end{wrapfigure}

A {\em $3$-fan}, denoted as $G_n$, is a $3$-colored outerplanar graph with $3n$ vertices ($n \geq 2$) and defined as follows. $G_n$ consists of a simple cycle formed by $n$ vertices of color $0$, followed (in the counterclockwise order) by $n$ vertices of color $1$, followed by $n$ vertices of color $2$.
The vertex of color $i$ adjacent in the cycle to a vertex of color $i-1$ (indices taken modulo $3$) is denoted as $v_i$. 
Also, in $G_n$ every vertex colored $i$ is adjacent to $v_i$ ($i=0,1,2$) and vertices $v_0,v_1,v_2$ form a $3$-cycle of $G_n$. See, e.g. Fig.~\ref{fi:Gn}. 
The following theorem has been proved in~\cite{DBLP:journals/jgaa/GiacomoDLMTW08}.

\begin{theorem}\emph{\cite{DBLP:journals/jgaa/GiacomoDLMTW08}}\label{th:gn}
Let $h$ be a positive integer and let $G_n$ be a $3$-fan for $n \geq 79 h^3$, and let $S_n$ be an alternating point set compatible with $G_n$. In every $3$-colored point-set embedding of $G_n$ on $S_n$ there is one edge with more than $h$ bends.
\end{theorem} 

The forest of stars that we use to establish our lower bound is called a {\em $3$-sky} and is  denoted by $F_n$. It consists of three stars $T_0$, $T_1$, $T_2$ such that: (i) each $T_i$ ($i=0,1,2$) has $n$ vertices ($n \geq 2$); (ii) all the vertices of each $T_i$ ($i=0,1,2$) have the same color $i$. 

Let $\Gamma_n$ be a point-set embedding of $F_n$ on $S_n$. An {\em uncrossed triplet} of $\Gamma_n$ is a triplet $p_i, p_{i+1}, p_{i+2}$ of points of $S_n$ such that, when moving along $CH(S_n)$ in clockwise order, no edge of $\Gamma_n$ crosses $CH(S_n)$ between $p_i$ and $p_{i+1}$ and between $p_{i+1}$ and $p_{i+2}$. A triplet is \emph{crossed $k$ times} if the total number of times that $CH(S_n)$ is crossed by some edges between $p_i$ and $p_{i+1}$ and between $p_{i+1}$ and $p_{i+2}$ is $k$. A {\em leaf triplet} of $\Gamma_n$ is a triplet of $S_n$ whose points represent leaves of $F_n$. Analogously, a \emph{root triplet} is a triplet of $S_n$ whose points represent the three roots of $F_n$. The following lemma establishes the first relationship between the curve complexity of some special types of $3$-colored point-set embeddings of $F_n$ and those of a $3$-fan $G_n$. 

\begin{lemma}\label{le:uncrossed-triplet}
Let $F_n$ be a $3$-sky, $S_n$ be an alternating point set compatible with $F_n$, and $\Gamma_n$ be a $3$-colored topological point-set embedding of $F_n$ on $S_n$. If $\Gamma_n$ has an uncrossed leaf triplet and each edge of $\Gamma_n$ crosses $CH(S_n)$ at most $b$ times, then the $3$-fan $G_n$ has a $3$-colored topological point-set embedding on $S_n$ such that each edge crosses $CH(S_n)$ at most $3b + 2$ times.
\end{lemma}
\begin{proof}
We show how to use $\Gamma_n$ to construct a topological point-set embedding of the $3$-fan $G_n$ on $S_n$ with at most $3b + 2$ crossings of $CH(S_n)$ per edge.

Let $p_j, p_{j+1}, p_{j+2}$ be an uncrossed leaf triplet. Every point of the triplet represents a leaf of a different star (because they have different color).  Denote by $q_i$ the point of $\Gamma_n$ representing the root of $T_i$ ($i=0,1,2$) and denote by $e_i$ the edge connecting $q_i$ to $p_{j+i}$. The idea is to connect the three points $q_0,q_1,q_2$ with a $3$-cycle that does not cross any existing edges. For each edge $e_i$ we draw two curves that from $q_i$ run very close to $e_i$ until they reach $CH(S_n)$. The two curves are drawn on the same side of $e_i$ such that they are consecutive in the circular order of the edges around $q_i$ (see Fig.~\ref{fi:cycle0} for an illustration). These two curves do not intersect any existing edges and cross $CH(S_n)$ the same number of times as $e_i$. The six drawn curves are now suitably connected to realize a cycle $C$ connecting $q_0,q_1,q_2$. Depending on which side the various curves reach $CH(S_n)$, the connections are different. However in all cases we can connect two curves to form a single edge by crossing $CH(S_n)$ at most two additional times and without violating planarity (see Fig.~\ref{fi:cycle1},~\ref{fi:cycle2}, and~\ref{fi:cycle3}). Thus, we have added to $\Gamma_n$ three edges $e'_i$ connecting $q_i$ to $q_{i+1}$ (indices taken modulo $3$), each crossing $CH(S_n)$ at most $2b+2$ times. Also, since the two curves that follow an edge $e_i$ are both drawn on the same side of $e_i$, the cycle $C$ does not have any vertices inside. Notice that, depending on the case with respect to the connection of the curves, $q_0$, $q_1$, and $q_2$ appear along $C$ either in the clockwise or in the counterclockwise order. W.l.o.g. we assume that the clockwise order is  $q_0$, $q_1$, and $q_2$.

\begin{figure}[tb]
	\centering
	\begin{minipage}[b]{.24\textwidth}
		\centering
		\includegraphics[width=\textwidth]{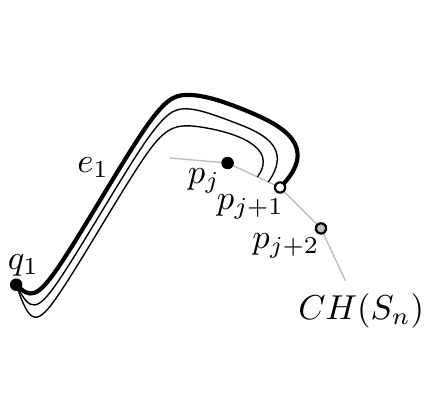}
		\subcaption{}\label{fi:cycle0}
	\end{minipage}
	\hfil
	\begin{minipage}[b]{.24\textwidth}
		\centering
		\includegraphics[width=\textwidth]{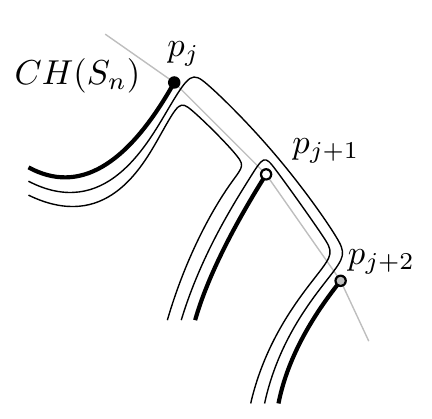}
		\subcaption{}\label{fi:cycle1}
	\end{minipage}
	\hfil
	\begin{minipage}[b]{.24\textwidth}
		\centering
		\includegraphics[width=\textwidth]{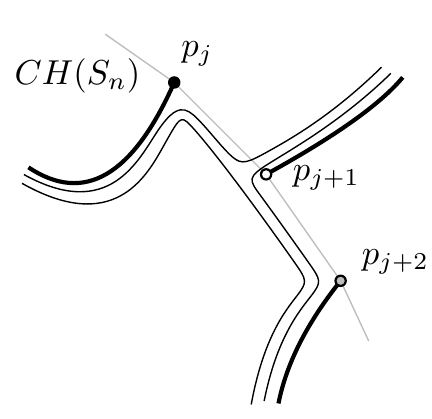}
		\subcaption{}\label{fi:cycle2}
	\end{minipage}
	\hfil
	\begin{minipage}[b]{.24\textwidth}
		\centering
		\includegraphics[width=\textwidth]{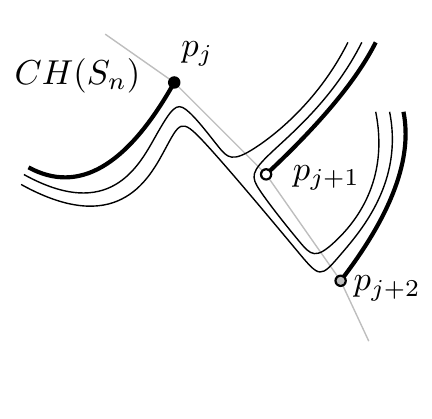}
		\subcaption{}\label{fi:cycle3}
	\end{minipage}	
	\caption{Insertion of a cycle connecting $q_0$, $q_1$ and $q_2$. (a) Drawing of the curves following the edge $e_1$. (b), (c), and (d) Connection of the six curves to form a cycle.\label{fi:cycle}}
\end{figure}

The obtained drawing is not yet a topological point-set embedding of $G_n$ because the cycle $C'$ connecting all the vertices is missing. We first add edges connecting leaves of the same color. Let $e'_i=\overline{e}_0,\overline{e}_1,\dots,\overline{e}_{n-2}=e'_{i-1}$ be the edges incident to $q_i$ in the circular order around $q_i$ (this is the counterclockwise order under our assumption that $q_0$, $q_1$, and $q_2$ are located in the clockwise order along $C$). We add an edge between the leaf of $\overline{e}_j$ and the leaf of $\overline{e}_{j+1}$ (for $j=0,1,\dots,n-3$) as follows. Starting from the leaf of $\overline{e}_j$, we draw a curve following the edge $\overline{e}_j$ until we arrive very close to $q_{i}$ and then we follow $\overline{e}_{j+1}$ until we reach the leaf of $\overline{e}_{j+1}$. The added edges do not cross any existing edges and cross $CH(S_n)$ a number of times equal to the number of times that $\overline{e}_j$ crosses $CH(S_n)$ plus the number of times that $\overline{e}_{j+1}$ crosses $CH(S_n)$, so at most $2b$.

It remains to add the edges of $C'$ connecting vertices of different colors. There are three such edges and they connect vertex $v_i$ ($i=0,1,2$) of $G_n$ to a vertex of color $i+1$ (indices taken modulo $3$). We add an edge connecting $q_i$ to a leaf of color $i+1$ as follows. Let $e''_{i+1}$ be the edge incident to $q_{i+1}$ that follows $e'_i$ in the clockwise order around $q_{i+1}$ (this is an edge connecting $q_{i+1}$ to a leaf of color $i+1$). Starting from $q_i$ we draw a curve following the edge $e'_i$ until we arrive very close to $q_{i+1}$ and then we follow $e''_{i+1}$ until we reach the leaf of $e''_{i+1}$. The constructed curve connects $q_{i}$ to a leaf  
of color $i+1$ and does not cross any existing edge. It crosses $CH(S_n)$ at most the number of times that $e'_{i+1}$ crosses $CH(S_n)$ (that is $2b+2$) plus the number of times that $e''_{i+1}$ crosses $CH(S_n)$ (that is $b$). Thus the total number of crossing of $CH(S_n)$ is at most $3b+2$.\qed 
\end{proof}

The next two lemmas explain how to obtain a $3$-colored topological book embedding that satisfies Lemma~\ref{le:uncrossed-triplet}.

\begin{lemma}\label{le:crossed-triplet}
Let $F_n$ be a $3$-sky, $S_n$ be an alternating point set compatible with $F_n$, and $\Gamma_n$ be a $3$-colored topological point-set embedding of $F_n$ on $S_n$ with a root triplet. If $\Gamma_n$ has a leaf triplet $\tau$ that is crossed $c$ times ($c < n$) and each edge crosses $CH(S_{n})$ at most $b$ times, then there exists a $3$-sky $F_{n'}$ which is a  subgraph of $F_n$ and an alternating point set $S_{n'}$ which is a subset of $S_n$ such that: (i) $n' \geq n-c$; (ii) there exists a $3$-colored topological point-set embedding $\Gamma_{n'}$ of $F_{n'}$ on $S_{n'}$ such that each edge crosses $CH(S_{n'})$ at most $b+1$ times; (iii) $\tau$ is an uncrossed leaf triplet of $\Gamma_{n'}$.
\end{lemma}

\begin{lemma}\label{le:root-triplet}
Let $F_n$ be a $3$-sky, $S_n$ be an alternating point set compatible with $F_n$, and $\Gamma_n$ be a $3$-colored topological point-set embedding of $F_n$ on $S_n$. If each edge of $\Gamma_n$ crosses $CH(S_n)$ at most $b$ times, then there exists a $3$-sky $F_{n'}$ which is a subgraph of $F_n$ and an alternating point set $S_{n'}$ which is a  subset of $S_n$ such that: (i) $n' \geq \frac{n}{3}-3$; (ii) there exists a $3$-colored topological point-set embedding $\Gamma_{n'}$ of $F_{n'}$ on $S_{n'}$ such that each edge crosses $CH(S_{n'})$ at most $b+2$ times; (iii) $\Gamma_{n'}$ has a root triplet.
\end{lemma}

\begin{lemma}\label{le:lower-bound}
Let $h$ be a positive integer, $F_n$ be a $3$-sky for $n = 520710h^3$, and $S_n$ be an alternating point set compatible with $F_n$. In every $3$-colored point-set embedding of $F_n$ on $S_n$ there exist at least $h^2$ edges with more than $h$ bends.
\end{lemma}
\begin{proof}[sketch]
Let $F_{n_i}$, $i=1,2,\dots,h^2$, be a $3$-sky for $n_i=520689 h^3+21h \cdot i$ and let $S_{n_i}$ be an alternating point set compatible with $F_{n_i}$. We prove by induction on $i$ that in every $3$-colored point-set embedding of $F_{n_i}$ on $S_{n_i}$ there exist $i$ edges with more than $h$ bends.
Notice that for $i=h^2$, we have $n_i=n$.

\textbf{Base case: $i=1$:} We have to prove that in any $3$-colored point-set embedding of $F_{n_1}$ on $S_{n_1}$ with $n_1=520689 h^3+21h$, there exists one edge with more than $h$ bends. 
Suppose as a contradiction that there exists a $3$-colored point-set embedding $\Gamma_{n_1}$ of $F_{n_1}$ on $S_{n_1}$ with curve complexity $h$. $\Gamma_{n_1}$ is also a $3$-colored topological point-set embedding of $F_{n_1}$ on $S_{n_1}$ such that each edge crosses $CH(S_{n_1})$ at most $2h$ times (each edge consists of at most $h+1$ segments). By Lemma~\ref{le:root-triplet} there exists a $3$-colored point-set embedding $\Gamma_{n'}$ of a $3$-sky $F_{n'}$ on an alternating point set $S_{n'}$ such that: (i) $n' \geq \frac{n_1}{3}$; (ii) each edge of $\Gamma_{n'}$ crosses $CH(S_{n'})$ at most $2h+2$ times; (iii) $\Gamma_{n'}$ has a root triplet. 

Since each edge of $\Gamma_{n'}$ crosses $CH(S_{n'})$ at most $2h+2$ times and there are  $3(n'-1)\geq n_1-3$ edges in total, there are at most $(2h+2)(n_1-3)$ crossings of $CH(S_{n_1})$ in total. The number of leaf triplets in $\Gamma_{n'}$ is $n'-1 \geq \frac{n_1}{3}-1$. It follows that there is at least one leaf triplet $\tau$ crossed at most $\frac{3(2h+2)(n_1-3)}{(n_1-3)} = 6h+6 \leq 7h$ times. By Lemma~\ref{le:crossed-triplet} there exists a $3$-colored point-set embedding $\Gamma_{n''}$ of a $3$-sky $F_{n''}$ on an alternating point set $S_{n''}$ such that: (i) $n'' \geq n'- 7h$; (ii) each edge of $\Gamma_{n''}$ crosses $S_{n''}$ at most $2h+3$ times; (iii) $\tau$ is uncrossed. By Lemma~\ref{le:uncrossed-triplet}, the $3$-fan $G_{n''}$ has a $3$-colored topological point-set embedding on $S_{n''}$ such that each edge crosses $CH(S_{n''})$ at most $6h+11$ times and by Lemma~\ref{le:convex} a $3$-colored point set embedding with curve complexity at most $12h+23$. On the other hand, since $n_1 = 520689 h^3+21h$, we have that $n'' \geq n'-7h \geq \frac{n_1}{3}-7h = \frac{520689 h^3 +21 h}{3}-7h=\frac{520689}{3}h^3 \geq \frac{520689}{3}h^3=79 (13h)^3$ and by Theorem~\ref{th:gn}, in every $3$-colored point-set embedding of $G_{n''}$ on $S_{n''}$ at least one edge that has more than $13h$ bends -- a contradiction.

\textbf{Inductive step: $i>1$.} We have to prove that in any $3$-colored point-set embedding of $F_{n_i}$ on $S_{n_i}$ with $n_i=520689 h^3+21h \cdot i$, there exist $i$ edges with more than $h$ bends.

We first prove that there exists at least one edge with more than $h$ bends. Suppose as a contradiction that there exists a $3$-colored point-set embedding $\Gamma_{n_i}$ of $F_{n_i}$ on $S_{n_i}$ with curve complexity $h$. With the same reasoning as in the base case, there would exist  a $3$-colored point set embedding with curve complexity at most $12h+23$ of a  $3$-fan $G_{n''}$, with $n'' \geq \frac{n_i}{3}-7h$. Since $n_i = 520689 h^3+21h \cdot i$, we have that $n'' \geq \frac{n_i}{3}-7h = \frac{520689 h^3 +21 h\cdot i}{3}-7h=\frac{520689}{3}h^3+7h(i-1) \geq \frac{520689}{3}h^3=79 (13h)^3$ and by Theorem~\ref{th:gn}, in every $3$-colored point-set embedding of $G_{n''}$ on $S_{n''}$ at least one edge has more than $13h$ bends -- again a contradiction.

This proves that there is at least one edge $e$ crossed more than $h$ times. We now remove this edge and the whole triplet that contains the point representing the leaf of $e$. We then arbitrarily remove $21h-1$ triplets. The resulting drawing is a $3$-colored point-set embedding $\Gamma_{n'''}$ of $F_{n'''}$ on $S_{n'''}$ for $n'''=n_{i-1}$. By induction, it contains $i-1$ edges each having more than $h$ bends. It follows that $\Gamma_{n_i}$ has $i$ edges each having more than $h$ bends. Since for $i=h^2$ we have $n_i=n$, the statement follows.
\qed  
\end{proof}
%

\begin{theorem}\label{th:lower-bound}
For sufficiently large $n$, there exists a $3$-colored forest $F_n$ consisting of three monochromatic stars with $n$ vertices and a $3$-colored point set $S_n$ in convex position compatible with $F_n$ such that any $3$-colored point-set embedding of $F_n$ on $S_n$ has $\Omega(n^{\frac{2}{3}})$ edges having $\Omega(n^{\frac{1}{3}})$ bends.
\end{theorem}

We conclude this section with some results deriving from Theorem~\ref{th:lower-bound} and/or related to it. Firstly, Theorem~\ref{th:lower-bound} extends the result of Theorem~\ref{th:gn} since it implies that a $3$-colored point set embedding of $G_n$ may require $\Omega(n^{\frac{2}{3}})$ edges with $\Omega(n^{\frac{1}{3}})$ bends each. Moreover, the result of Theorem~\ref{th:lower-bound} implies an analogous result for a $k$-colored forest of at least three stars for every $k \geq 3$. In particular, when $k=n$ we have the following result that extends the one by Pach and Wenger~\cite{DBLP:journals/gc/PachW01}. 

\begin{corollary}
Let $F$ be a forest of three $n$-vertex stars. Every planar drawing of $F$ with vertices at fixed vertex locations has $\Omega(n^{\frac{2}{3}})$ edges with $\Omega(n^{\frac{1}{3}})$ bends each.
\end{corollary}

One may wonder whether the lower bound of Theorem~\ref{th:lower-bound} also holds when the number of colors or the number of stars is less than three. However, it is immediate to see that this is not the case, i.e., the following theorem holds.

\begin{theorem}
Let $F$ be a $k$-colored forest of $h$ stars and $S$ be a set of points compatible with $F$. If $\max\{k,h\}=2$ then $F$ has a $k$-colored point-set embedding on $S$ with curve complexity at most $2$.
\end{theorem}  

Since a caterpillar can be regarded as a set of stars whose roots are connected in a path, the lower bound of Theorem~\ref{th:lower-bound} also holds for caterpillars. This answers an open problem in~\cite{DBLP:journals/tcs/BadentGL08} about the curve complexity of $k$-colored point-set embeddings of trees for $k \geq 3$. Note that $O(1)$ curve complexity for $2$-colored outerplanar graphs has been proved in~\cite{DBLP:journals/jgaa/GiacomoDLMTW08}.

\begin{corollary}\label{co:caterpillars}
For sufficiently large $n$, a $3$-colored point-set embedding of a $3$-colored caterpillar may require $\Omega(n^{\frac{2}{3}})$ edges having $\Omega(n^{\frac{1}{3}})$ bends.
\end{corollary}

\section{Point-set Embeddings of Paths and Caterpillars}\label{se:upper}

In the light of Corollary~\ref{co:caterpillars}, one may ask whether there exist subclasses of $3$-colored caterpillars for which constant curve complexity can be guaranteed. In this section we first prove that this is the case for $3$-colored paths and then we extend the result to $3$-colored caterpillars whose leaves all have the same color.   

Based on Lemma~\ref{le:top-pse}, we prove that a $3$-colored path $P$ has a topological book embedding consistent with $seq(S)$ and having a constant number of spine crossings. Namely, we first remove the vertices and points of one color from $P$ and $S$, obtaining a $2$-colored path $P'$ and a compatible $2$-colored point set $S'$. Next, we construct a topological book embedding $\gamma_{P'}$ of $P'$ consistent with $seq(S')$  with at most two spine crossings per edge and with suitable properties. Then we use such properties to reinsert the third color and obtain a topological book embedding of $P$ consistent with $seq(S)$. 

$P'$ and $\sigma'=seq(S')$ can be regarded as two binary strings of the same size where one color is represented by bit $0$ and the other one by bit $1$. $P'$ and $\sigma'$ are \emph{balanced} if the number of $0$'s ($1$'s, resp.) in $P'$ equals the number of $0$'s ($1$'s, resp.) in $\sigma'$. 
$P'$ and $\sigma'$ are a \emph{minimally balanced pair} if there does not exist a prefix of $P'$ and a corresponding prefix of $\sigma'$ that are balanced. 

\begin{lemma}\label{le:twin-chunks}
Let $P$ and $\sigma$ be a minimally balanced pair of length $k>1$. Let $b_j(P)$  denote the $j$-th bit of $P$ and $b_j(\sigma)$ denote the $j$-th bit of $\sigma$. Then  $b_1(P) \neq b_k(P)$, $b_k(P) = b_1(\sigma)$, and $b_1(P) = b_k(\sigma)$.
\end{lemma}

Let $\Gamma$ be a topological book embedding, $\ell$ be the spine of $\Gamma$, and $p$ be a point of $\ell$ (possibly representing a vertex). We say that $p$ is visible from above (below) if the vertical ray with origin at $p$ and lying in the top (bottom) page does not intersect any edge of $\Gamma$. We say that the segment $\overline{pq}$ is \emph{visible from above (below)} if each point $r$ in the segment is visible from above (below). Let $u$ and $v$ be two vertices of $\Gamma$ that are consecutive along the spine $\ell$, we say that segment $\overline{uv}$ is \emph{accessible} if it contains a segment that is visible from below.
A vertex $v$ of $\Gamma$ is \emph{hook visible} if there exists a segment $\overline{pq}$ of the spine such that $\overline{pq}$ is visible from below and for any point $r$ of $\overline{pq}$ we can add an edge in the top page of $\Gamma$ connecting $v$ with $r$ without crossing any other edges of $\Gamma$ (see Fig.~\ref{fi:hook-visibility}); $\overline{pq}$ is the \emph{access interval for vertex $v$}. If the access interval is to the right (left) of $v$ we say that $v$ is \emph{hook visible from the right (left)}. 

\begin{figure}[tb]
	\centering
	\begin{minipage}[b]{.24\textwidth}
		\centering
		\includegraphics[width=\textwidth]{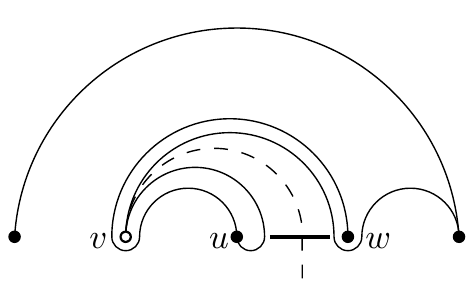}
		\subcaption{}\label{fi:hook-visibility}
	\end{minipage}
	\begin{minipage}[b]{.24\textwidth}
		\centering
		\includegraphics[width=\textwidth]{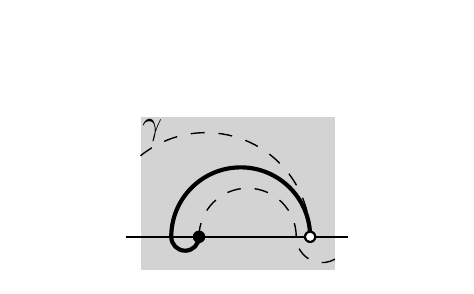}
		\subcaption{}\label{fi:basecase}
	\end{minipage}	
	\begin{minipage}[b]{.24\textwidth}
		\centering
		\includegraphics[width=\textwidth]{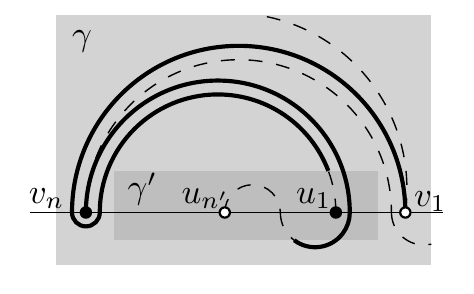}
		\subcaption{}\label{fi:case1}
	\end{minipage}		
	\begin{minipage}[b]{.24\textwidth}
		\centering
		\includegraphics[width=\textwidth]{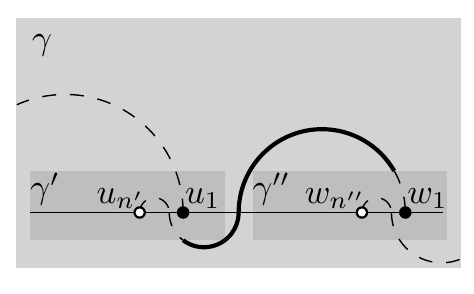}
		\subcaption{}\label{fi:case2}
	\end{minipage}		
	\caption{\label{fi:boh} (a) Illustration of the hook visibility property. The bold segment is the access interval. (b)-(d) Proof of Lemma~\ref{le:2-colored-book-embedding}: (b) Base cas; (c) Case 1; (d) Case 2.}	
\end{figure}

\begin{lemma}\label{le:2-colored-book-embedding}
Let $P$ be a $2$-colored path and $\sigma$ be a $2$-colored sequence compatible with $P$. Path $P$ admits a topological book embedding $\gamma$ consistent with $\sigma$ and with the following properties: 
\begin{inparaenum}
	\item[(a)] Every edge of $\gamma$ crosses the spine at least once and 	at most twice.
	\item[(b)] For any two vertices $u$ and $v$ that are consecutive along the spine $\ell$ of $\gamma$, segment $\overline{uv}$ is accessible from below.
	\item[(c)] Every spine crossing is visible from below.
	\item[(d)] The first vertex $v_1$ of $P$ is visible from above; the last vertex $v_n$ of $P$ is hook visible from the right; to the right of its access interval there is only one vertex and no spine crossing.
\end{inparaenum}
\end{lemma}
\begin{proof}
We prove the statement by induction on the length $n$ of $P$ (and of $\sigma$). If $n=1$ the statement trivially holds. If $n=2$ we draw the unique edge of $P$ with one spine crossing immediately to the left of the leftmost vertex in the drawing (see Fig.~\ref{fi:basecase}). Also in this case the statement holds. Suppose that $n>2$ and that the statement holds for every $k<n$. We distinguish between two cases.

\textbf{Case 1: $P$ and $\sigma$ are a minimally balanced pair.} By Lemma~\ref{le:twin-chunks} the first vertex of $P$ has the same color as the last element of $\sigma$, the last element of $P$ has the same color as the first element of $\sigma$ and these two colors are different. This means that by removing the first and the last elements from both $P$ and $\sigma$, we obtain a new $2$-colored path $P'$ of length $n-2$ and a new $2$-colored sequence $\sigma'$ compatible with $P'$. By induction, $P'$ admits a topological book embedding $\gamma'$ consistent with $\sigma'$ and satisfying properties (a)--(d). To create a topological book embedding of $P$ consistent with $\sigma$, we add a point $p_1$ before all the points of $\gamma'$, whose color is the same as the last vertex $v_n$ of $P$, and a point $p_2$ after all points of $\gamma'$, whose color is the same as the first vertex $v_1$ of $P$. Vertex $v_1$ is mapped to $p_2$ and vertex $v_n$ is mapped to $p_1$. We connect $v_1$ to the first vertex $u_1$ of $P'$ with an edge incident to $p_2$ from above, crossing the spine once immediately before $p_1$ and once immediately after $p_1$ and arriving to $u_1$ from above (by property (d), $u_1$ is visible from above). We then connect the last vertex $u_{n'}$ of $P'$ to $v_n$. Since $u_{n'}$ is hook visible by property (d), we connect it to $v_n$ with an edge that starting from $u_{n'}$ reaches the access interval of $u_{n'}$, crosses the spine between the last vertex of $\gamma'$ and $p_2$ and reaches $p_1$ from above. As shown in Fig.~\ref{fi:case1} the two edges $(v_1,u_1)$ and $(u_{n'},v_n)$ can be added without creating any crossing. Property (a) holds by construction. About properties (b) and (c), we added two arcs in the bottom page. The first one connects a point immediately before $v_n$ and a point immediately after it, so the segment of the spine between $v_n$ and its following vertex is accessible from below; also, the addition of this arc does not change the accessibility of the spine crossing of $\gamma'$. The second arc added in the bottom page connects a point $q$ in the access interval of $u_{n'}$ with a point immediately after $u_1$; by property (d) of $\gamma'$ there is no vertex or spine crossings between $q$ and $u_1$. Thus the segments connecting $u_1$ to its preceding and to its following vertices are visible from below and property (b) holds; furthermore the addition of this arc does not change the accessibility to existing spine crossings. Since the new created spine crossings are visible from below, property (c) also holds. It is immediate to see that also property (d) holds; see for example Fig.~\ref{fi:case1}.             

\textbf{Case 2: $P$ and $\sigma$ are not a minimally balanced pair.} In this case there exists a prefix (i.e. a subpath) $P'$ of $P$ and a corresponding prefix $\sigma'$ of $\sigma$ that are balanced. $P'$ is $2$-colored path and $\sigma'$ is a $2$-colored sequence compatible with $P'$ and their length is less than $n$. By induction, $P'$ admits a topological book embedding $\gamma'$ consistent with $\sigma'$ and statisfying properties (a)--(d). On the other hand, $P''=P\setminus P'$ is also a $2$-colored path and $\sigma''=\sigma \setminus \sigma'$ is a $2$-colored sequence consistent with $P''$. Thus, $P''$ also admits a topological book embedding $\gamma''$ consistent with $\sigma''$ and statisfying properties (a)--(d). Since the last vertex $u_n$ of $P'$ is hook visible in $\gamma'$ and the first vertex $w_1$ of $P''$ is visible from above in $\gamma''$, the two vertices can be connected with an edge that crosses the spine twice (see Fig.~\ref{fi:case2}), thus creating a topological book embedding $\gamma$ of $P$ consistent with $\sigma$. 
Property (a) holds by construction. The only arc added in the bottom page connects a point $q$ in the access interval of $u_n$ and a point $q'$ immediately after the first vertex $u_1$ of $P'$. By property (d) of $\gamma'$ there is no vertex or spine crossing between $q$ and $u_1$ and between $u_1$ and $q'$, thus properties (b) and (c) hold for $\gamma$. Property (d) holds because it holds for $\gamma'$ and $\gamma''$.
\qed         
\end{proof}



\begin{lemma}\label{le:3-colored-book-embedding}
A $3$-colored path admits a topological book embedding with at most two spine crossings per edge consistent with any compatible  $3$-colored sequence.
\end{lemma}
\begin{proof}[sketch]
Let $c_2$ be a color distinct from the colors of the end-vertices of $P$. Let $v_1,v_2,\dots,v_k$ be a maximal subpath of $P$ colored $c_2$. Let $u_1$ and $u_2$ be the vertices along $P$ before $v_1$ and after $v_k$, respectively. 
We replace the subpath $u_1,v_1,v_2,\dots,v_k,u_2$ with an edge $(u_1,u_2)$. We do the same for every maximal subpath colored $c_2$. Let $P'$ be the resulting $2$-colored path and $\sigma'$ be the $2$-colored sequence obtained from $\sigma$ by removing all elements of color $c_2$.


By Lemma~\ref{le:2-colored-book-embedding}, $P'$ admits a topological book embedding $\gamma'$  consistent with $\sigma'$ that satisfies properties (a), (b), (c) and (d). We add to $\gamma'$ a set $Q$ of points colored $c_2$ to represent the removed vertices that will be added back. These points must be placed so that the sequence of colors along the spine coincides with $\sigma$. By property (b) of $\gamma'$ all these points can be placed so that they are accessible from below. 
We now have to replace some edges of $P'$ with paths of vertices colored $c_2$. 
Let $(u_1,u_2)$ be an edge that has to be replaced by a path $u_1,v_1,v_2,\dots,v_k,u_2$. For each vertex $v_i$ to be added ($i=1,2,\dots,k$) we add an \emph{image point} to the drawing. The image points are added as follows. By property (a), the edge $(u_1,u_2)$ crosess the spine at least once. Let $\chi$ be the point where $(u_1,u_2)$ crosses the spine for the first time when going from $u_1$ to $u_2$. By property (c) $\chi$ is visible from below. This means there is a segment $s$ of $\ell$ with $\chi$ as an endpoint that is visibile from below. We place $k-1$ image points $p_1,p_2,\dots,p_{k-1}$ inside this segment, while $\chi$ is the $k$-th image point $p_k$ (it is the leftmost if $s$ is to the left of $\chi$, while it is the rightmost if $s$ is to the right of $\chi$). The first arc of the edge $(u_1,u_2)$ is replaced by an arc connecting $u_1$ to $p_1$. Each image point $p_i$ is connected to the $p_{i+1}$ ($i=1,2,k-1$) by means of an arc in the top page. Finally, the last image point $p_k$ is already connected to $u_2$ by means of the remaining part of the original edge $(u_1,u_2)$. Notice that the edge $(u_1,p_1)$ does not cross the spine, and the same is true for any edge $(p_i,p_{i+1})$, while the edge $(p_k,u_2)$ crosses the spine at most once (the original edge had at most two spine crossing one of which was at $\chi=p_k$). We have replaced the edge $(u_1,u_2)$ with a path $\pi=\langle u_1, p_1,p_2,\dots,p_k, u_2 \rangle$ with $k+1$ edges, as needed. However, the points representing the intermediate vertices of this path are not the points of the set $Q$. The idea then is to ``connect'' the image points to the points of $Q$. To this aim, we add matching edges in the bottom page between the image points and the points of $Q$. Since both the points of $Q$ and the image points are visible from below, these matching edges do not cross any other existing edge. Moreover, by using a simple brackets matching algorithm, we can add the matching edges so that they do not cross each other. Finally the matching edges can be used to create the actual path that represent $u_1,v_1,v_2,\dots,v_k,u_2$. 
%
%
\qed
\end{proof}

%
%

The following theorem is a consequence of Lemmas~\ref{le:top-pse} and~\ref{le:3-colored-book-embedding}.

\begin{theorem}\label{th:path}
Every $3$-colored path admits a $3$-colored point-set embedding with curve complexity at most $5$ on any compatible $3$-colored point set.
\end{theorem}

Theorem~\ref{th:path} can be extended to a subclass of $3$-colored caterpillars.

\begin{theorem}\label{th:caterpillar}
Every $3$-colored caterpillar with monochromatic leaves admits a $3$-colored point-set embedding with curve complexity at most $5$ on any compatible $3$-colored point set.
\end{theorem}

The above results motivate the study of $4$-colored graphs, in particular a natural question is whether $4$-colored paths admit point-set embedding on any set of points with constant curve complexity. 

\begin{theorem}
Let $P$ be a $4$-colored path with $n$ vertices and let $S$ be a $4$-colored point set compatible with $P$. If the first $h \geq 2$ vertices along $P$ only have two colors and the remaining $n-h$ only have the other two colors, then $P$ has a $4$-colored point-set embedding on $S$ with curve complexity at most $5$.
\end{theorem}

\section{Open Problems}


Motivated by the results of this paper we suggest the following open problems:
\begin{inparaenum}[(i)]
\item Investigate whether the lower bound of Theorem~\ref{th:lower-bound} is tight. We recall that an upper bound of $O(n)$ holds for all $n$-colored planar graphs~\cite{DBLP:journals/gc/PachW01}. 
\item Study whether constant curve complexity can always by guaranteed for $4$-colored paths. 
\item Characterize the $3$-colored caterpillars that admit a $3$-colored point-set embedding with constant curve complexity on any given set of points.
\end{inparaenum} 

\bibliography{biblio}
\bibliographystyle{splncs03}

\newpage

\appendix

\section*{Appendix}

\section{Omitted proofs from Section~\ref{se:preliminaries}}

\newcounter{app}
\setcounter{app}{\value{lemma}}
\setcounter{lemma}{0}

\begin{lemma}\label{le:convex}
Let $G$ be a $k$-colored graph, let $S$ be a $k$-colored one-sided convex point set compatible with $G$. If $G$ has a topological $k$-colored point-set embedding on $S$ such that each edge crosses $CH(S)$ at most $b$ times, then $G$ admits a $k$-colored point-set embedding on $S$ with at most $2b+1$ bends per edge.  
\end{lemma}
\begin{proof}
Replace each crossing between an edge $e$ and $CH(S)$ with a division vertex on $e$ and connect any two (real or division) vertices that are consecutive along $CH(S)$ with an edge if they are not yet connected. The resulting graph is a $k$-colored Hamiltonian augmentation consistent with $seq(S)$ (since $S$ is one-sided convex the circular order of the colors along $CH(S)$ coincides with $seq(S)$) and with at most $b$ division vertices per edge. By Theorem~\ref{th:hamiltonian} $G$ admits a $k$-colored point-set embedding on $S$ with at most $2b+1$ bends per edge.\qed
\end{proof}

\section{Omitted proofs from Section~\ref{se:lower}}

\setcounter{lemma}{3}

\begin{lemma}\label{le:crossed-triplet}
Let $F_n$ be a $3$-sky, $S_n$ be an alternating point set compatible with $F_n$, and $\Gamma_n$ be a $3$-colored topological point-set embedding of $F_n$ on $S_n$ with a root triplet. If $\Gamma_n$ has a leaf triplet $\tau$ that is crossed $c$ times ($c < n$) and each edge crosses $CH(S_{n})$ at most $b$ times, then there exists a subgraph $F_{n'}$ of $F_n$ and a subset $S_{n'}$ of $S_n$ such that: (i) $n' \geq n-c$; (ii) there exists a $3$-colored topological point-set embedding $\Gamma_{n'}$ of $F_{n'}$ on $S_{n'}$ such that each edge crosses $CH(S_{n'})$ at most $b+1$ times; (iii) $\tau$ is an uncrossed leaf triplet of $\Gamma_{n'}$.
\end{lemma}
\begin{proof}
The idea is to make $\tau$ uncrossed by removing all edges that cross it. In order to keep the set of points alternating, for each removed edge $e$ we will remove the whole triplet that contains the point representing the leaf of $e$. Suppose first that none of the triplets to be removed coincide with $\tau$ (i.e., no edge that crosses $\tau$ has an endpoint in $\tau$). In this case, we remove all edges that cross $\tau$ and all the triplets containing their leaves. Notice that, since $\Gamma_n$ has a root triplet, no root is removed. Since we always remove triplets, the final drawing has the same number of points for each color. Thus, such a drawing is a $3$-colored topological point-set embedding $\Gamma_{n'}$ of a $3$-sky $F_{n'}$ on an alternating point set $S_{n'}$ with an uncrossed triplet $\tau$. The number of triplets removed is at most $c$, and therefore $n' \geq n-c$. Finally, since we have only removed edges, the number of times that an edge crosses $CH(S_{n'})$ remains $b$.

Suppose now that at least one edge that crosses $\tau$ has an endpoint in $\tau$. Denote by $\overline{E}$ the set of these edges. We first remove all edges that cross $\tau$ but are not in $\overline{E}$. Then we redraw the edges of $\overline{E}$ so that they do not cross $\tau$. To this aim, consider a closed curve $C$ around $\tau$ that only intersects edges that are incident to or cross $\tau$. 
%
%
We cut the three edges incident to the points of $\tau$ (these three edges include those in $\overline{E}$, possibly coinciding with them) at the point where they cross $C$ for the first time. We then connect the three points on the curve $C$ to the points of $\tau$. One can observe that this can always be done by adding to each edge at most one crossing of $CH(S_{n})$ (see Figure~\ref{fi:cross} for an illustration). Also in this case the resulting drawing is a $3$-colored topological point-set embedding $\Gamma_{n'}$ of a $3$-sky $F_{n'}$ on an alternating point set $S_{n'}$ with an uncrossed triplet $\tau$. Again, the number of triplets removed is at most $c$, and therefore $n' \geq n-c$. 
\qed
\end{proof}

\begin{figure}[tb]
	\centering
	\begin{minipage}[b]{.3\textwidth}
		\centering
		\includegraphics[width=\textwidth]{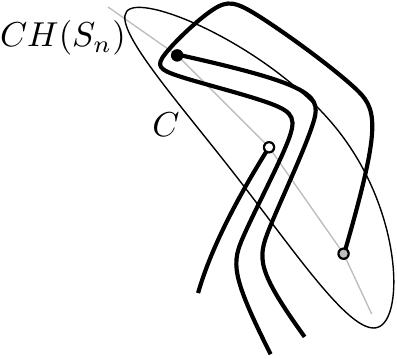}
		\subcaption{}\label{fi:cross1}
	\end{minipage}
	\hfil
	\begin{minipage}[b]{.3\textwidth}
		\centering
		\includegraphics[width=\textwidth]{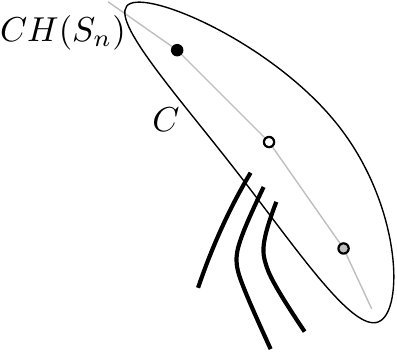}
		\subcaption{}\label{fi:cross2}
	\end{minipage}
	\hfil
	\begin{minipage}[b]{.3\textwidth}
		\centering
		\includegraphics[width=\textwidth]{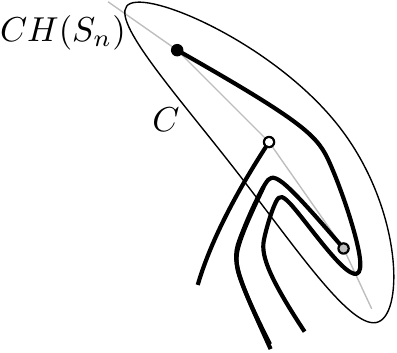}
		\subcaption{}\label{fi:cross3}
	\end{minipage}	
	\caption{\label{fi:cross} Re-drawing of the edges incident to a leaf triplet to make it uncrossed.}
\end{figure}

\begin{lemma}\label{le:root-triplet}
Let $F_n$ be a $3$-sky, $S_n$ be an alternating point set compatible with $F_n$, and $\Gamma_n$ be a $3$-colored topological point-set embedding of $F_n$ on $S_n$. If each edge of $\Gamma_n$ crosses $CH(S_n)$ at most $b$ times, then there exists a subgraph $F_{n'}$ of $F_n$ and a subset $S_{n'}$ of $S_n$ such that: (i) $n' \geq \frac{n}{3}-3$; (ii) there exists a $3$-colored topological point-set embedding $\Gamma_{n'}$ of $F_{n'}$ on $S_{n'}$ such that each edge crosses $CH(S_{n'})$ at most $b+2$ times; (iii) $\Gamma_{n'}$ has a root triplet.
\end{lemma}
\begin{proof}
If $\Gamma_n$ has a root triplet, then $\Gamma_{n'}$ coincides with $\Gamma_n$ and the thesis holds. Suppose then that $\Gamma_{n}$ does not have a root triplet. Consider the triplets that contain the roots. They are at most three of them. Denote them as $\tau_1$, $\tau_2$ and $\tau_3$. Denote by $T_{i,i+1}$ (for $i=0,1,2$, indices taken modulo $3$) be the set of triplets that are encountered between $\tau_i$ and $\tau_i+1$ when moving clockwise along $CH(S_{n})$. The number of triplets in $T_{0,1}$, $T_{1,2}$, and $T_{2,0}$ is $n-3$ and thus at least one of these three sets has $\frac{n-3}{3}$ triplets. We remove all the triplets of the other two sets. Furthermore we remove the points of the triplets containing the roots that do not represent the roots. This removes at most six extra vertices. The three roots are now consecutive. If they form a triplet, i.e., they are colored $0$, $1$, $2$ in the clockwise order, then the drawing obtained after the removal is the desired $\Gamma_{n'}$. In $\Gamma_{n'}$, there are at least $\frac{n-3}{3}$ triplets plus the root triplet, hence $n'\geq\frac{n-3}{3}+1=\frac{n}{3}$.

\begin{figure}[tb]
	\centering
	\begin{minipage}[b]{.35\textwidth}
		\centering
		\includegraphics[width=\textwidth,page=3]{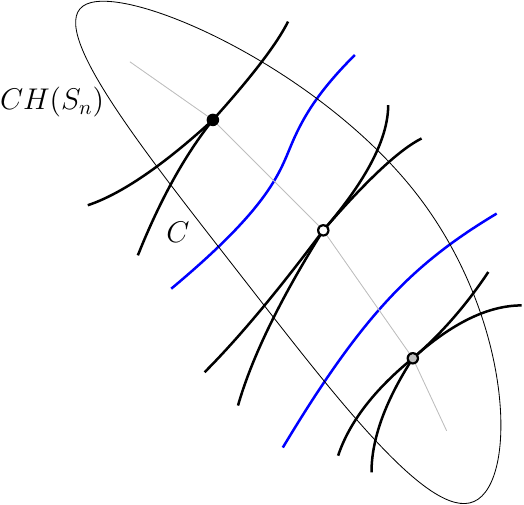}
		\subcaption{}\label{fi:root-triplet}
	\end{minipage}
	\hfil
	\begin{minipage}[b]{.35\textwidth}
		\centering
		\includegraphics[width=\textwidth,page=2]{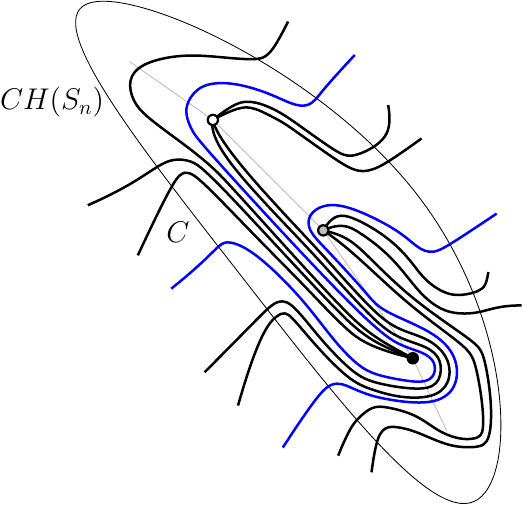}
		\subcaption{}\label{fi:root-triplet}
	\end{minipage}
	\caption{Re-drawing of the edges incident to and crossing the root triplet, in order to change the seqeunce of colors.\label{fi:cycle}}
\end{figure}

If the three roots do not form a triplet, we locally modify the drawing to transform them into a triplet. To this aim, we consider a closed curve $C$ around the three roots that is crossed only by the edges incident to the three roots and by the edges that cross $CH(S_n)$ between them. It is possible to reroute the edges crossing $C$ so as to guarantee that the clockwise order of the root colors is $0$, $1$, $2$. This can be done by adding to each edge at most two crossing of $CH(S_{n})$ (see Figure~\ref{fi:root-triplet} for an illustration). After the rerouting, we are in the same situation as in the previous case and the resulting drawing is the desired $\Gamma_{n'}$ with $n'\geq \frac{n}{3}$.
\qed 
\end{proof}

\begin{lemma}\label{le:lower-bound}
Let $h$ be a positive integer, $F_n$ be a $3$-sky for $n = 520710h^3$, and $S_n$ be an alternating point set compatible with $F_n$. In every $3$-colored point-set embedding of $F_n$ on $S_n$ there exist at least $h^2$ edges with more than $h$ bends.
\end{lemma}
\begin{proof}
For $i=1,2,\dots,h^2$ let $F_{n_i}$ be a $3$-sky for $n_i=520689 h^3+21h \cdot i$ and let $S_{n_i}$ be an alternating point set compatible with $F_{n_i}$. We prove by induction on $i$ that in every $3$-colored point-set embedding of $F_{n_i}$ on $S_{n_i}$ there exist $i$ edges with more than $h$ bends.
Notice that for $i=h^2$, we have $n_i=n$.

\textbf{Base case: $i=1$:} We have to prove that in any $3$-colored point-set embedding of $F_{n_1}$ on $S_{n_1}$ with $n_1=520689 h^3+21h$, there exists one edge with more than $h$ bends. 
Suppose as a contradiction that there exists a $3$-colored point-set embedding $\Gamma_{n_1}$ of $F_{n_1}$ on $S_{n_1}$ with curve complexity $h$. $\Gamma_{n_1}$ is also a $3$-colored topological point-set embedding of $G_{n_1}$ on $S_n$ such that each edge crosses $CH(S_{n_1})$ at most $2h$ times (each edge consists of at most $h+1$ segments). By Lemma~\ref{le:root-triplet} there exists a $3$-colored point-set embedding $\Gamma_{n'}$ of a $3$-sky $F_{n'}$ on an alternating point set $S_{n'}$ such that: (i) $n' \geq \frac{n_1}{3}$; (ii) each edge of $\Gamma_{n'}$ crosses $CH(S_{n'})$ at most $2h+2$ times; (iii) $\Gamma_{n'}$ has a root triplet. 

Since each edge of $\Gamma_{n'}$ crosses $CH(S_{n'})$ at most $2h+2$ times and there are  $3(n'-1)\geq n_1-3$ edges in total, there are at most $(2h+2)(n_1-3)$ crossings of $CH(S_{n_1})$ in total. The number of leaf triplets in $\Gamma_{n'})$ is $n'-1 \geq \frac{n_1}{3}-1$. It follows that there is at least one leaf triplet $\tau$ crossed at most $\frac{3(2h+2)(n_1-3)}{(n_1-3)} = 6h+6 \leq 7h$ times. By Lemma~\ref{le:crossed-triplet} there exists a $3$-colored point-set embedding $\Gamma_{n''}$ of a $3$-sky $F_{n''}$ on an alternating point set $S_{n''}$ such that: (i) $n'' \geq n'- 7h$; (ii) each edge of $\Gamma_{n''}$ crosses $S_{n''}$ at most $2h+3$ times; (iii) $\tau$ is uncrossed. By Lemma~\ref{le:uncrossed-triplet}, the $3$-fan $G_{n''}$ has a $3$-colored topological point-set embedding on $S_{n''}$ such that each edge crosses $CH(S_{n''})$ at most $6h+11$ times and by Lemma~\ref{le:convex} a $3$-colored point set embedding with curve complexity at most $12h+23$. On the other hand, since $n_1 = 520689 h^3+21h$, we have that $n'' \geq n'-7h \geq \frac{n_1}{3}-7h = \frac{520689 h^3 +21 h}{3}-7h=\frac{520689}{3}h^3 \geq \frac{520689}{3}h^3=79 (13h)^3$ and by Theorem~\ref{th:gn}, in every $3$-colored point-set embedding of $G_{n''}$ on $S_{n''}$ at least one edge that has more than $13h$ bends -- a contradiction.

\textbf{Inductive step: $i>1$.} We have to prove that in any $3$-colored point-set embedding of $F_{n_i}$ on $S_{n_i}$ with $n_i=520689 h^3+21h \cdot i$, there exist $i$ edges with more than $h$ bends.

We first prove that there exists at least one edge with more than $h$ bends. As in the base case, suppose as a contradiction that there exists a $3$-colored point-set embedding $\Gamma_{n_i}$ of $F_{n_i}$ on $S_{n_i}$ with curve complexity $h$. By Lemma~\ref{le:root-triplet} there exists a $3$-colored point-set embedding $\Gamma_{n'}$ of a $3$-sky $F_{n'}$ on an alternating point set $S_{n'}$ such that: (i) $n' \geq \frac{n_i}{3}$; (ii) each edge of $\Gamma_{n'}$ crosses $CH(S_{n'})$ at most $2h+2$ times; (iii) $\Gamma_{n'}$ has a root triplet.

With the same argument used in the base case, we conclude that in $\Gamma_{n'}$ there is at least one leaf triplet $\tau$ crossed at most $7h$ times. By Lemma~\ref{le:crossed-triplet} there exists a $3$-colored point-set embedding $\Gamma_{n''}$ of a $3$-sky $F_{n''}$ on an alternating point set $S_{n''}$ such that: (i) $n'' \geq n'- 7h$; (ii) each edge of $\Gamma_{n''}$ crosses $S_{n''}$ at most $2h+3$ times; (iii) $\tau$ is uncrossed. By Lemmas~\ref{le:uncrossed-triplet} and~\ref{le:convex} the $3$-fan $G_{n''}$ has a $3$-colored point set embedding with curve complexity at most $12h+23$. On the other hand, since $n_i = 520689 h^3+21h \cdot i$, we have that $n'' \geq n'-7h \geq \frac{n_i}{3}-7h = \frac{520689 h^3 +21 h\cdot i}{3}-7h=\frac{520689}{3}h^3+7h(i-1) \geq \frac{520689}{3}h^3=79 (13h)^3$ and by Theorem~\ref{th:gn}, in every $3$-colored point-set embedding of $G_{n''}$ on $S_{n''}$ at least one edge has more than $13h$ bends -- again a contradiction.

This proves that there is at least one edge $e$ crossed more than $h$ times. We now remove this edge and the whole triplet that contains the point representing the leaf of $e$. We then arbitrarily remove $21h-1$ triplets. The resulting drawing is a $3$-colored point-set embedding $\Gamma_{n'''}$ of $F_{n'''}$ on $S_{n'''}$ for $n'''=n_{i-1}$. By induction, it contains $i-1$ edges each having more than $h$ bends. It follows that $\Gamma_{n_i}$ has $i$ edges each having more than $h$ bends. Since for $i=h^2$ we have $n_i=n$, the statement follows.
\qed  
\end{proof}

\newcounter{app2}
\setcounter{app2}{\value{theorem}}

\setcounter{theorem}{3}

\begin{theorem}
Let $F$ be a $k$-colored forest of $h$ stars and $S$ be a set of points compatible with $F$. If $\max\{k,h\}=2$ then $F$ has a $k$-colored point-set embedding on $S$ with curve complexity at most $2$.
\end{theorem}  
\begin{proof}
Suppose first that $k=2$. In this case we connect the roots of the stars and obtain a $2$-colored caterpillar, which admits a $2$-colored point-set embedding on any $2$-colored set of points with curve complexity two~\cite{DBLP:journals/ijfcs/GiacomoLT06}.
Suppose now that $h=2$. It is immediate to see that a star admits a book embedding on one page (i.e. a topological book embedding such that all edges are in the same page) for any chosen ordering of the vertices. Thus, using a distinct page for each of the two stars, we can obtain a topological book embedding consistent with $seq(S)$, for any $S$, without spine crossings. By Lemma~\ref{le:top-pse} $F$ has a $k$-colored point-set embedding on $S$ with curve complexity $1$.\qed    
\end{proof}

\section{Omitted proofs from Section~\ref{se:upper}}

\setcounter{lemma}{6}
\begin{lemma}\label{le:twin-chunks}
Let $P$ and $\sigma$ be a minimally balanced pair of length $k>1$. Let $b_j(P)$  denote the $j$-th bit of $P$ and let $b_j(\sigma)$ denote the $j$-th bit of $\sigma$. Then  $b_1(P) \neq b_k(P)$, $b_k(P) = b_1(\sigma)$, and $b_1(P) = b_k(\sigma)$.
\end{lemma}
\begin{proof}
We first observe that $b_1(P) \neq b_1(\sigma)$ as otherwise $b_1(P)$ and $b_1(\sigma)$ would be balanced prefixes of $P$ and of $\sigma$, respectively, thus contradicting the fact that $P$ and $\sigma$ are a minimally balanced pair. Also, $b_k(P) \neq b_k(\sigma)$ or else the prefixes obtained by removing $b_k(P)$ from $P$ and $b_k(\sigma)$ from $\sigma$ would be balanced, again contradicting the fact that $P$ and $\sigma$ are a minimally balanced pair. 

Assume now by contradiction that $b_1(P) = b_k(P)$, which would also imply that $b_k(P) \neq b_1(\sigma)$ (since $b_1(P) \neq b_1(\sigma)$). Denote by $\Delta_l$ the difference between the number of $0$'s in the string $b_1(P)b_2(P) \ldots b_l(P)$ and the number of $0$'s in the string $b_1(\sigma)b_2(\sigma) \ldots b_l(\sigma)$ (for $l=1,2,\dots,k$). Assume w.l.o.g. that $b_1(P)=0$, then $b_1(\sigma)=1$ and therefore $\Delta_1=1$. On the other hand, since $P$ and $\sigma$ are balanced, $\Delta_k=0$. Since $b_k(P) \neq b_k(\sigma)$, $\Delta_{k-1} \neq 0$ and since $b_1(P) = b_k(P)$, it is $b_k(P)=0$, which implies $\Delta_{k-1}=-1$. The value of $\Delta_l$ changes by at most one unit (positively or negatively) when the index $l$ is increased by one unit. Thus, $\Delta_1=1$ and $\Delta_{k-1}=-1$ imply that there must exist an index $j$, with $1 < j < k-1$, where $\Delta_j=0$. But this would imply the existence of a prefix of $P$ and a corresponding prefix of $\sigma$ that are balanced -- again a contradiction.       

It remains to show that $b_1(P) = b_k(\sigma)$. But this is an immediate consequence of the fact that $b_1(P) \neq b_k(P)$ and $b_k(P) \neq b_k(\sigma)$.\qed
\end{proof}

\setcounter{theorem}{5}
\begin{theorem}\label{th:caterpillar}
Let $C$ be a $3$-colored caterpillar with monochromatic leaves and let $S$ be a 3-colored point set compatible with $C$. $P$ has a $3$-colored point-set embedding on $S$ with curve complexity at most $5$.
\end{theorem}
\begin{proof}
As in the case of $3$-colored paths, we first remove from $C$ the vertices of one color. In particular, we remove the vertices with the color of the leaves. Let $c_2$ be this color. The leaves are simply removed from $C$ while the subpaths of the backbone\footnote{The path obtained from a caterpillar by removing all its leaves is usually called \emph{spine}. To avoid confusion with the spine of the book embedding, we call it \emph{backbone}.} of $C$ whose color is $c_2$ are replaced by dummy edges as in the case of paths. Let $P'$ be the resulting $2$-colored path and let $\sigma'$ be the $2$-colored sequence obtained from $\sigma$ by removing all elements of color $c_2$.

By Lemma~\ref{le:2-colored-book-embedding}, $P'$ admits a topological book embedding $\gamma'$  consistent with $\sigma'$ that satisfies the properties (a), (b), (c), (d). As in the proof of Lemma~\ref{le:3-colored-book-embedding} we add to $\gamma'$ a set $Q$ of points colored $c_2$ placed so that the sequence of colors along the spine coincides with $\sigma$. By property (b) of $\gamma'$ all these points can be placed so that they are accessible from below. 
We now have to remove the dummy edges of $P'$ and re-insert the vertices colored $c_2$. The approach is very similar to the one used for paths. The difference is that each dummy edge has to be replaced  with a subgraph that is a caterpillar.

\begin{figure}[tb]
	\centering
	\begin{minipage}[b]{.48\textwidth}
		\centering
		\includegraphics[width=\textwidth,page=1]{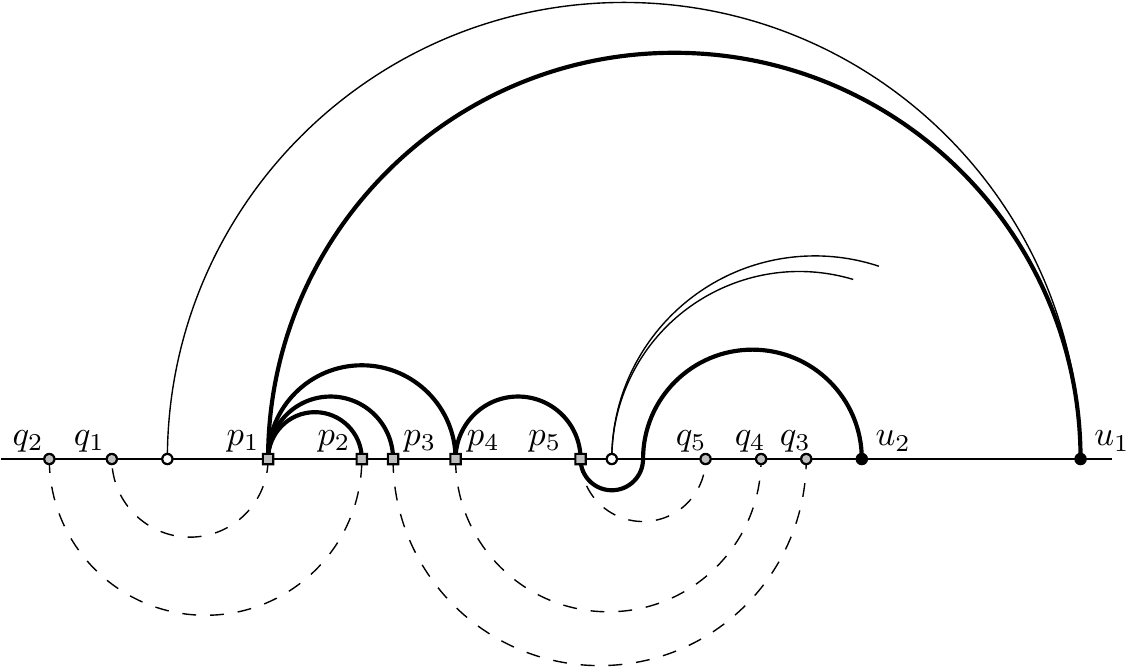}
		\subcaption{}\label{fi:matching-replacement-2}
	\end{minipage}
	\begin{minipage}[b]{.48\textwidth}
		\centering
		\includegraphics[width=\textwidth,page=2]{matching-replacement-2}
		\subcaption{}\label{fi:matching-replacement-2}
	\end{minipage}	
	\caption{\label{fi:matching-replacement-2} Addition of the vertices of color $c_2$ to create a topological book embedding of a $3$-colored caterpillar.}	
\end{figure}

Let $(u_1,u_2)$ be an edge that has to be replaced by a subgraph $C'$ with $k$ vertices (notice that $C'$ contains $u_1$ and $u_2$ as leaves attached to the first and last vertex of the backbone of $C'$). We add $k$ image points to the drawing as done in the proof of Lemma~\ref{le:3-colored-book-embedding} and, as in that proof, the first arc of the edge $(u_1,u_2)$ is replaced by an arc connecting $u_1$ to $p_1$, while the last image point $p_k$ is already connected to $u_2$ by means of the remaining part of the original edge $(u_1,u_2)$. Finally, let $\gamma''$ be $1$-page book embedding of $C'$ with the property that $u_1$ and $u_2$ are the first and the last vertex of $\gamma''$ (such an embedding always exists). We ``embed'' $\gamma''$ on the image points, placing all the edges on the top page. After all dummy edges have been replaced with $1$-page book embeddings that use image points, we add the matching edges as done in the proof of Lemma~\ref{le:3-colored-book-embedding} and use them to construct the final drawing. For vertices that belong to the backbone of $C$, the matching edge is replaced by two edges as in the case of paths, while for the vertices that are leaves the matching edge is not replaced (see Figure~\ref{fi:matching-replacement-2}). As for the case of paths, the final drawing is planar and each edge crosses the spine at most twice.\qed    
\end{proof}

\begin{lemma}\label{le:3-colored-book-embedding}
Let $P$ be a $3$-colored path and let $\sigma$ be a $3$-colored sequence compatible with $P$. Path $P$ admits a topological book embedding $\Gamma_P$ consistent with $\sigma$ such that every edge of $\Gamma_P$ crosses the spine at most twice.
\end{lemma}
\begin{proof}[sketch]
Let $c_0$ and $c_1$ be the colors of the end-vertices of $P$ (possibly $c_0=c_1$) and $c_2$ be the color distinct from $c_0$ and $c_1$. Let $v_1,v_2,\dots,v_k$ be a maximal subpath of $P$ colored $c_2$. Let $u_1$ be the vertex before $v_1$ along $P$ and let $u_2$ be the vertex after $v_k$ along $P$. Since the end-vertices of $P$ have color distinct from $c_2$, $u_1$ and $u_2$ always exist. We replace the subpath $u_1,v_1,v_2,\dots,v_k,u_2$ with an edge $(u_1,u_2)$. We do this replacement for every maximal subpath colored $c_2$. Let $P'$ be the resulting $2$-colored path and let $\sigma'$ be the $2$-colored sequence obtained from $\sigma$ by removing all elements of color $c_2$.

\begin{figure}[tb]
	\centering
	\begin{minipage}[b]{.48\textwidth}
		\centering
		\includegraphics[width=\textwidth,page=1]{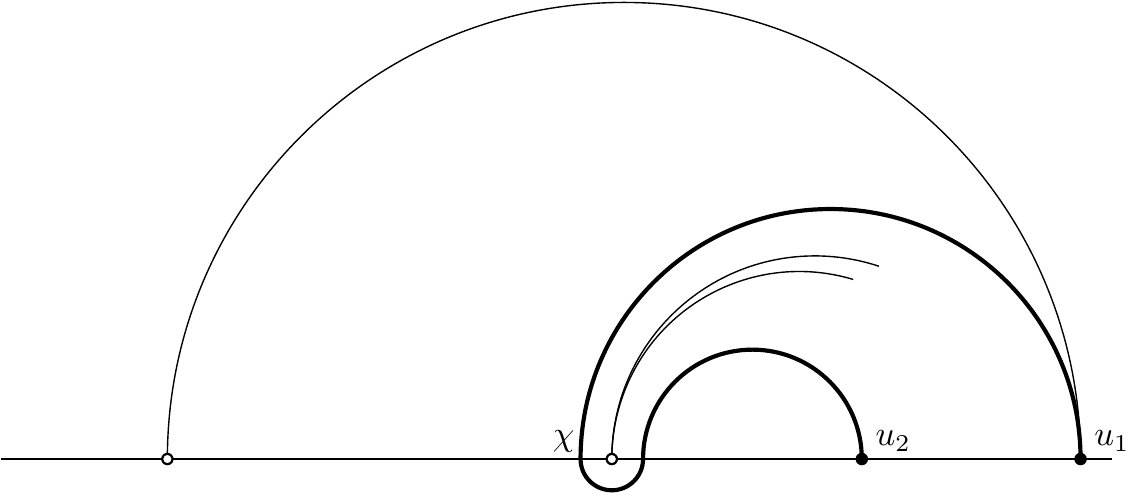}
		\subcaption{}\label{fi:matching-replacement}
	\end{minipage}
	\begin{minipage}[b]{.48\textwidth}
		\centering
		\includegraphics[width=\textwidth,page=2]{matching-replacement}
		\subcaption{}\label{fi:matching-replacement}
	\end{minipage}	
	\begin{minipage}[b]{.48\textwidth}
		\centering
		\includegraphics[width=\textwidth,page=3]{matching-replacement}
		\subcaption{}\label{fi:matching-replacement}
	\end{minipage}		
	\begin{minipage}[b]{.48\textwidth}
		\centering
		\includegraphics[width=\textwidth,page=4]{matching-replacement}
		\subcaption{}\label{fi:matching-replacement}
	\end{minipage}		
	\caption{\label{fi:matching-replacement} Addition of the vertices of color $c_2$ to create a topological book embedding of a $3$-colored path.}	
\end{figure}

By Lemma~\ref{le:2-colored-book-embedding}, $P'$ admits a topological book embedding $\gamma'$  consistent with $\sigma'$ that satisfies properties (a), (b), (c) and (d). We add to $\gamma'$ a set $Q$ of points colored $c_2$ to represent the removed vertices that will be added back. These points must be placed so that the sequence of colors along the spine coincides with $\sigma$. By property (b) of $\gamma'$ all these points can be placed so that they are accessible from below. 
We now have to replace some edges of $P'$ with paths of vertices colored $c_2$. Refer to Figure~\ref{fi:matching-replacement} for an illustration (a full example is shown in Figures~\ref{fi:2coloredtbe-2-3} and~\ref{fi:2coloredtbe-4-5}). Let $(u_1,u_2)$ be an edge that has to be replaced by a path $u_1,v_1,v_2,\dots,v_k,u_2$. For each vertex $v_i$ to be added ($i=1,2,\dots,k$) we add an \emph{image point} to the drawing. The image points are added as follows. By property (a), the edge $(u_1,u_2)$ crosess the spine at least once. Let $\chi$ be the point where $(u_1,u_2)$ crosses the spine for the first time when going from $u_1$ to $u_2$. By property (c) $\chi$ is visible from below. This means that there is a segment $s$ of $\ell$ with $\chi$ as an endpoint that is visibile from below. We place $k-1$ image points $p_1,p_2,\dots,p_{k-1}$ inside this segment, while $\chi$ is the $k$-th image point $p_k$ (it is the leftmost if $s$ is to the left of $\chi$, while it is the rightmost if $s$ is to the right of $\chi$). The first arc of the edge $(u_1,u_2)$ is replaced by an arc connecting $u_1$ to $p_1$; each image point $p_i$ is connected to the $p_{i+1}$ ($i=1,2,k-1$) by means of an arc in the top page; finally, the last image point $p_k$ is already connected to $u_2$ by means of the remaining part of the original edge $(u_1,u_2)$. Notice that the edge $(u_1,p_1)$ does not cross the spine, and the same is true for any edge $(p_i,p_{i+1})$, while the edge $(p_k,u_2)$ crosses the spine at most once (the original edge had at most two spine crossing one of which was at $\chi=p_k$). We have replaced the edge $(u_1,u_2)$ with a path $\pi=\langle u_1, p_1,p_2,\dots,p_k, u_2 \rangle$ with $k+1$ edges, as needed. However, the points representing the intermediate vertices of this path are not the points of the set $Q$. The idea then is to ``connect'' the image points to the points of $Q$. To this aim, we add matching edges in the bottom page between the image points and the points of $Q$. Since both the points of $Q$ and the image points are visible from below, these matching edges do not cross any other existing edge. Moreover, by using a simple brackets matching algorithm, we can add the matching edges so that they do not cross each other. We finally use the matching edges to create the actual path that will represent $u_1,v_1,v_2,\dots,v_k,u_2$. We ``cut'' the path $\pi$ at each image point $p_i$ and replace it with two consecutive points $p'_i$ and $p''_i$. The edge $(u_1,p_1)$ becomes an edge $(u_1,p'_1)$, each edge $(p_i,p_{i+1})$ becomes an edge $(p''_{i},p'_{i+1})$, and the edge $(p_k,u_2)$ becomes an edge $(p''_k,u_2)$. Denote by $q_i$ the point of $Q$ matched with $p_i$. For each $q_i$ we add two edges $(q_i,p'_i)$ and $(q_i,p''_i)$ following the matching edge $(p_i,q_i)$. In this way we have a path $u_1,p'_1,q_1,p''_1,p'_2,q_2,p''_2,\dots,p'_k,q_k,p''_k,u_2$. The edge $(q_k,p''_k)$ and the first arc of $(p''_k,u_2)$ are both in the same page; thus we remove $p''_k$ and connect $q_k$ directly to the second arc $(p''_k,u_2)$. The points $p'_i$ and $p''_i$ ($i=1,2,\dots,k$) do not represent vertices but spine crossings, while the points $q_i$ represent vertices of $P$. The resulting drawing is planar by construction. Also, each added edge crosses the spine at most twice. Namely, the first edge from $u_1$ to $q_1$ has a spine crossing at $p'_1$; each edge from $q_i$ to $q_{i+1}$ ($i=1,2,\dots,k$) has two spine crossings at $p''_i$ and a $p'_{i+1}$. Finally, the edge from $p_k$ to $u_2$ possibly has a spine crossing of the original edge $(u_1,u_2)$.\qed
\end{proof}

\begin{theorem}
Let $P$ be a $4$-colored path with $n$ vertices and let $S$ be a 3-colored point set compatible with $P$. If the first $h \geq 2$ vertices along $P$ only have two colors and the remaining $n-h$ only have the other two colors, then $P$ has a $4$-colored point-set embedding on $S$ with curve complexity at most $5$.
\end{theorem}
\begin{proof}
Removing the edge between the $h$-th vertex and the $(h+1)$-th vertex we obtain two $2$-colored paths $P_1$ and $P_2$. Let $\sigma_i$ ($i=1,2$) be the $2$-colored sequence obtained from $\sigma$ considering only the elements whose color is one of the colors of $P_i$. Let $\overline{P_1}$ be the ``reversed'' path of $P_1$. Path $P$ can be obtained by connecting the first vertex of $\overline{P_1}$ to the first vertex of $P_2$. By Lemma~\ref{le:2-colored-book-embedding} $\overline{P_1}$ (respectively $P_2$) has a topological book embedding $\overline{\gamma_1}$ (respectively $\gamma_2$) consistent with $\sigma_1$ (respectively $\sigma_2$) with the additional property that every segment between two consecutive vertices is visible from below (respectively above) and the first vertex is visible from above (respectively below). Thus, $\overline{\gamma}_1$ and $\gamma_2$ can be arranged together in a topological book embedding of $P_1 \cup P_2$ on $S$ consistent with $\sigma$ by intermixing their vertices according to the order defined by $\sigma$ and without creating any intersection between the two. Since the first vertex of $\overline{\gamma}_1$ is visible from above and the first one of $\gamma_2$ is visible from below, we can draw the unique missing edge with one spine crossing placed after any existing vertices and spine crossings.  \qed
\end{proof} 

\newpage

\section{Additional figures}

\begin{figure}[h!]
	\centering
	\includegraphics[width=\textwidth,page=1]{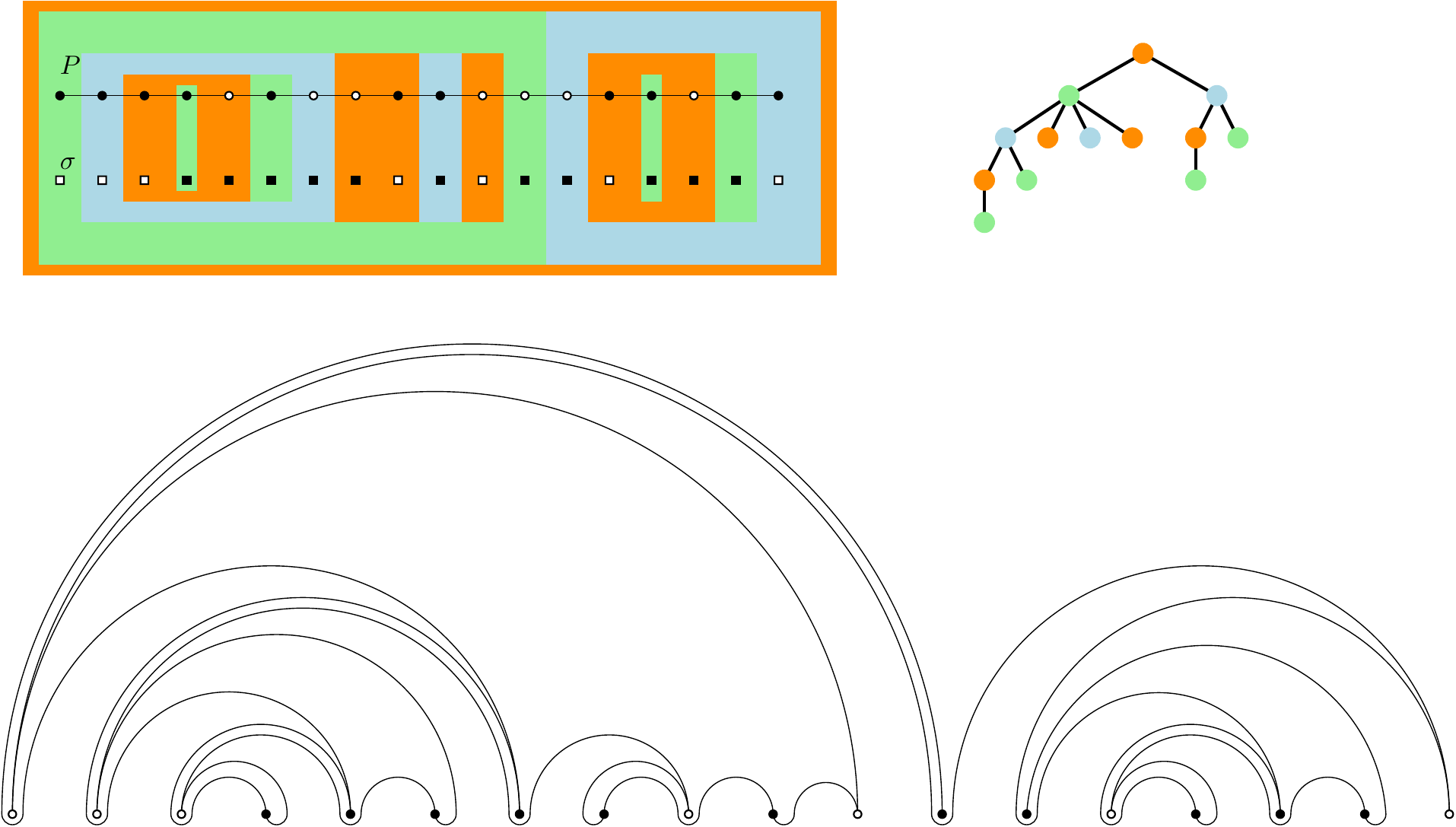}
	\caption{\label{fi:2coloredtbe}A topological book embedding of a $2$-colored path consistent with a given $2$-colored sequence constructed according to the algorithm described in the proof of Lemma~\ref{le:2-colored-book-embedding}. To help the reader, colored rectangles highlighting the input of the various recursive calls and the recursion tree are also shown.}	
\end{figure}

\begin{figure}[htbp]
	\centering
	\begin{minipage}[b]{\textwidth}
		\centering
		\includegraphics[width=\textwidth,page=2]{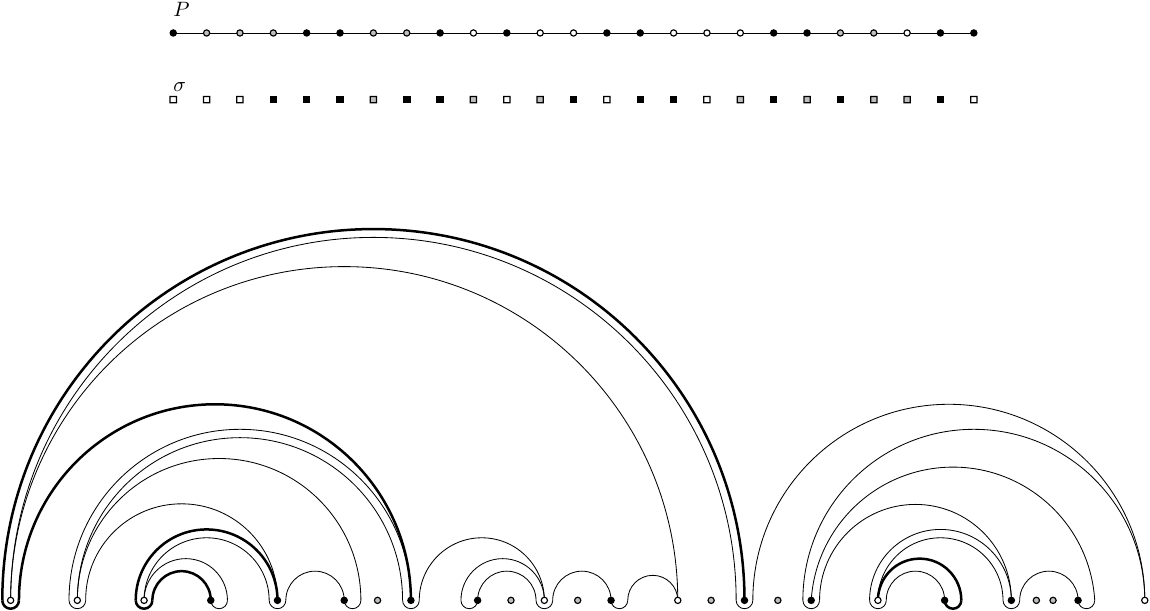}
		\subcaption{}\label{fi:2coloredtbe-2}
	\end{minipage}
	\begin{minipage}[b]{\textwidth}
		\centering
		\includegraphics[width=\textwidth,page=3]{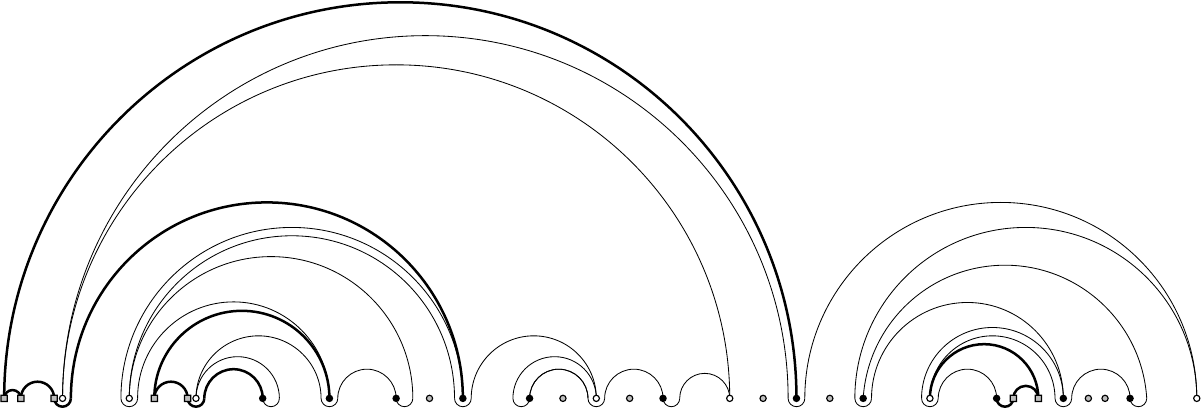}
		\subcaption{}\label{fi:2coloredtbe-3}
	\end{minipage}	
	\caption{\label{fi:2coloredtbe-2-3} Construction of a topological book embedding of a $3$-colored path consistent with a given $3$-colored sequence starting from the topological book embedding of Figure~\ref{fi:2coloredtbe} (1/2).}	
\end{figure}

\begin{figure}[htbp]
	\centering
	\begin{minipage}[b]{\textwidth}
		\centering
		\includegraphics[width=\textwidth,page=4]{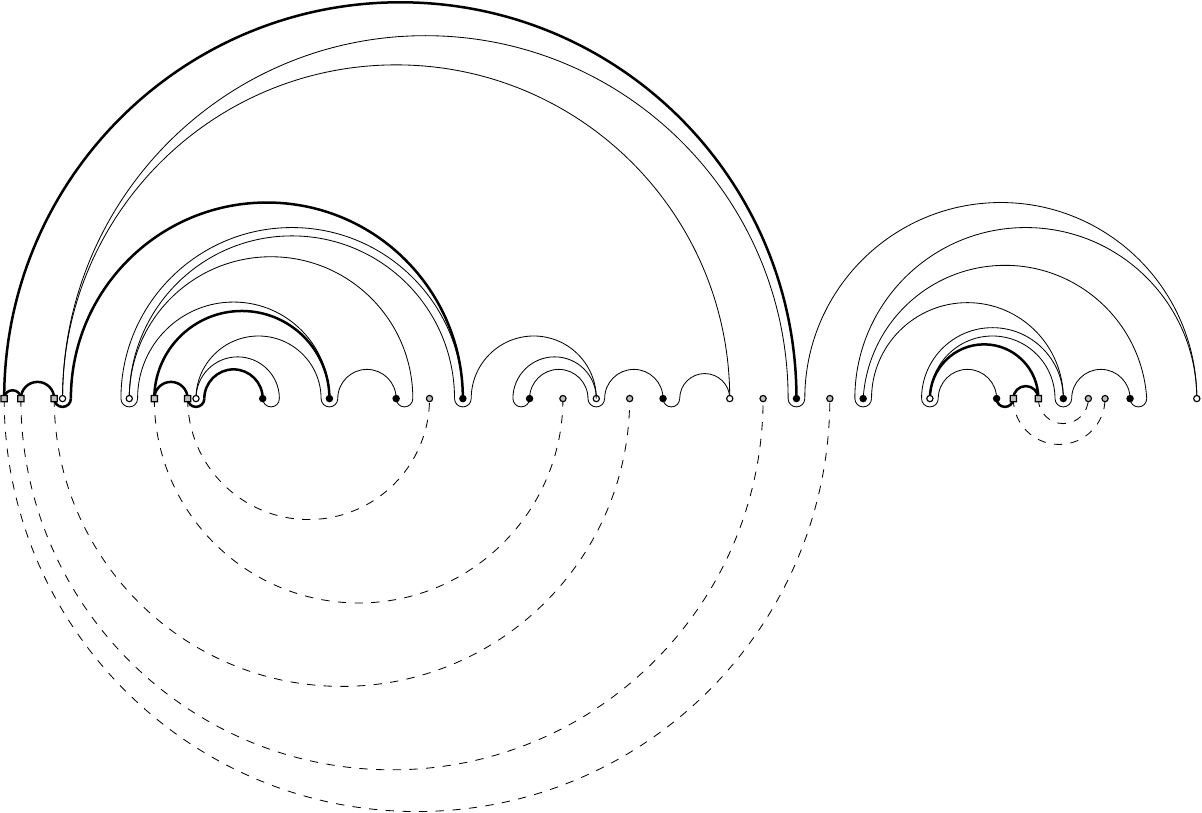}
		\subcaption{}\label{fi:2coloredtbe-4}
	\end{minipage}
	\begin{minipage}[b]{\textwidth}
		\centering
		\includegraphics[width=\textwidth,page=5]{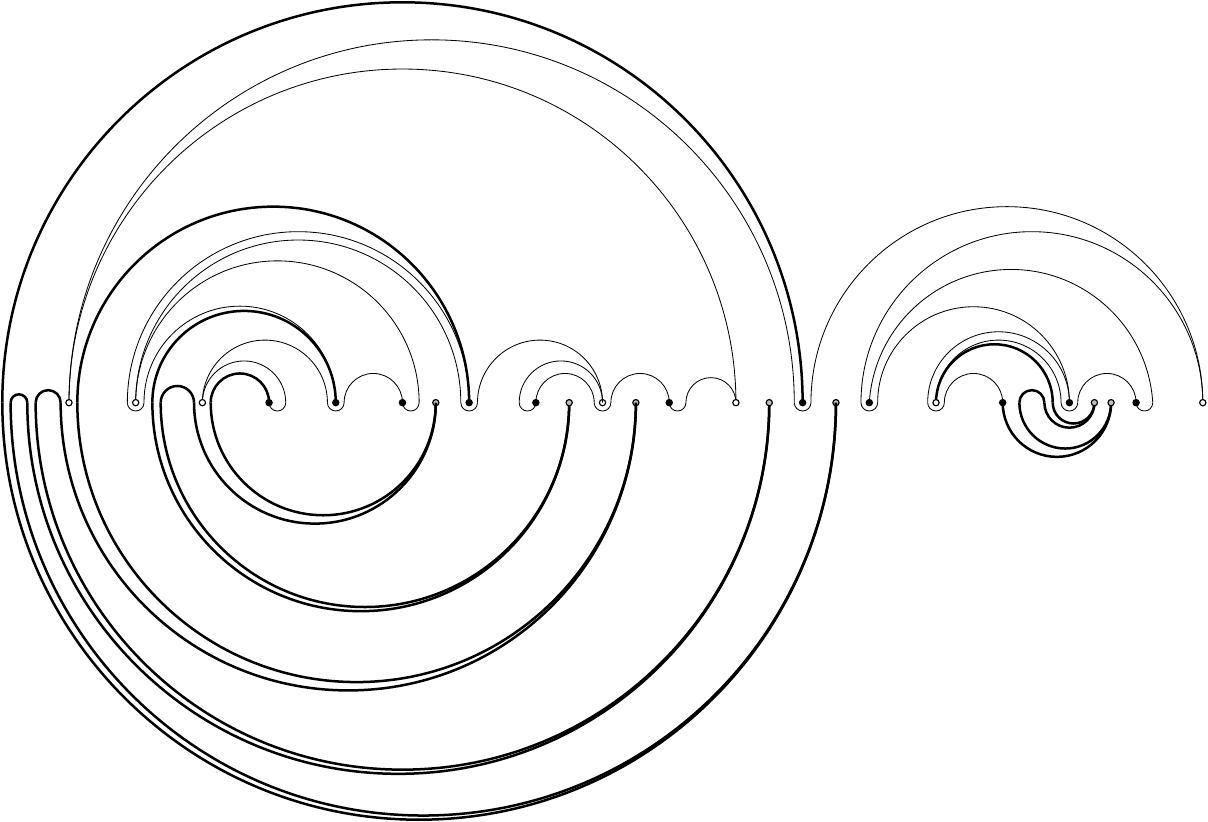}
		\subcaption{}\label{fi:2coloredtbe-5}
	\end{minipage}	
	\caption{\label{fi:2coloredtbe-4-5}Construction of a topological book embedding of a $3$-colored path consistent with a given $3$-colored sequence starting from the topological book embedding of Figure~\ref{fi:2coloredtbe} (2/2).}	
\end{figure}

\end{document}